\def\fullversion{1}

\ifnum\fullversion=1
\documentclass{article}[11pt,letterpaper]
\usepackage{amsmath, amssymb}
\usepackage{amsthm}
\usepackage{fullpage}
\usepackage{graphicx}

\usepackage{ifpdf,xspace}
\usepackage{shadow,shadethm,color}
\usepackage[compact]{titlesec}
\usepackage{algorithm, algorithmic}

\definecolor{Darkblue}{rgb}{0,0,0.4}
\definecolor{Brown}{cmyk}{0,0.81,1.,0.60}
\definecolor{Purple}{cmyk}{0.45,0.86,0,0}
\newcommand{\mydriver}{hypertex}
\ifpdf
 \renewcommand{\mydriver}{pdftex}
\fi
\usepackage[breaklinks,\mydriver]{hyperref}
\hypersetup{colorlinks=true,
            citebordercolor={.6 .6 .6},linkbordercolor={.6 .6 .6},%
citecolor=Darkblue,urlcolor=black,linkcolor=Darkblue,pagecolor=black}
\newcommand{\lref}[2][]{\hyperref[#2]{#1~\ref*{#2}}}


\newtheorem{theorem}{Theorem}[section]
\newtheorem{corollary}[theorem]{Corollary}
\newtheorem{observation}[theorem]{Observation}
\newtheorem{lemma}[theorem]{Lemma}

\newtheorem{definition}{Definition}

\newtheorem{remark}{Remark}

\newcommand{\Real}{\mathbb R}
\newcommand{\OPT}{{\sf OPT}}
\newcommand{\eps}{\varepsilon}

\newcommand{\one}{\mathbf{1}}

\renewcommand{\Pr}{{\bf Pr}}

\newcommand{\cost}{\textsf{cost}}
\newcommand{\E}{\mathbf{\mathbb{E}}}

\newcommand{\avg}{\mathop{\textrm{avg}}}

\newcommand{\ud}{\mathrm{d}}
\newcommand{\card}[1]{|#1|}
\newcommand{\ess}{\mathcal{S}}

\newcounter{note}[section]

\newcommand{\ktnote}[1]{}

\newcommand{\initOneLiners}{%
    \setlength{\itemsep}{0pt}
    \setlength{\parsep }{0pt}
    \setlength{\topsep }{0pt}
}
\newenvironment{OneLiners}[1][\ensuremath{\bullet}]
    {\begin{list}
        {#1}
        {\initOneLiners}}
    {\end{list}}
\newcommand{\sse}{\subseteq}
\newcommand{\ts}{\textstyle}
\newcommand{\whp}{\textbf{whp}\xspace}

\begin{document}

\title{Differentially Private Combinatorial Optimization}

\author{
Anupam Gupta
\and
Katrina Ligett
\and
Frank McSherry
\and
Aaron Roth
\and
Kunal Talwar
}

\maketitle
\thispagestyle{empty}
\begin{abstract}
  \bigskip Consider the following problem: given a metric space, some of
  whose points are ``clients,'' select a set of at most $k$ facility locations to
  minimize the average distance from the clients to their nearest facility.
  This is just the well-studied $k$-median problem, for which many
  approximation algorithms and hardness results are known. Note that the
  objective function encourages opening facilities in areas where there
  are many clients, and given a solution, it is often possible to get a
  good idea of where the clients are located. This raises the
  following quandary: what if the locations of the clients are sensitive
  information that we would like to keep private? \emph{Is it even
    possible to design good algorithms for this problem that preserve
    the privacy of the clients?}

  \medskip In this paper, we initiate a systematic study of algorithms
  for discrete optimization problems in the framework of differential
  privacy (which formalizes the idea of protecting the privacy of
  individual input elements). We show that many such problems indeed
  have good approximation algorithms that preserve differential privacy;
  this is even in cases where it is impossible to preserve cryptographic
  definitions of privacy while computing any non-trivial approximation
  to even the \emph{value} of an optimal solution, let alone the entire
  solution.

  \medskip Apart from the $k$-median problem, we consider the problems
  of vertex and set cover, min-cut, facility location, and
  Steiner tree, and give approximation algorithms and lower bounds for
  these problems.  We also consider the recently introduced submodular
  maximization problem, ``Combinatorial Public Projects'' (CPP), shown
  by Papadimitriou et al.  \cite{PSS08} to be inapproximable to
  subpolynomial multiplicative factors by any efficient and
  \emph{truthful} algorithm. We give a differentially private (and hence
  approximately truthful) algorithm that achieves a logarithmic additive
  approximation.
\end{abstract}

\newpage

\section{Introduction}
\label{sec:introduction}

Consider the following problems:
\begin{OneLiners}
\item Assign people using a social network  to one of
  two servers so that most pairs of friends are assigned to the same
  server.
\item Open some number of HIV treatment centers so that the
  average commute time for patients is small.
\item Open a small number of drop-off centers for undercover agents so
  that each agent is able to visit some site convenient to her (each
  providing a list of acceptable sites).
\end{OneLiners}

The above problems can be modeled as instances of well-known
combinatorial optimization problems: respectively the minimum cut
problem, the $k$-median problem, and the set cover problem.  Good
heuristics have been designed for these problems, and hence they may be
considered well-studied and solved.  However, in the above scenarios and
in many others, the input data (friendship relations, medical history,
agents' locations) represent sensitive information about individuals.
Data privacy is a crucial design goal, and it may be vastly preferable
to use a private algorithm that gives somewhat suboptimal solutions to a
non-private optimal algorithm.  This leads us to the following central
questions: \emph{Given that the most benign of actions possibly leaks
  sensitive information, how should we design algorithms for the above
  problems?  What are the fundamental trade-offs between the utility of
  these algorithms and the privacy guarantees they give us?}

The notion of privacy we consider in this paper is that of {\em
  differential privacy}. Informally, differential privacy guarantees
that the distribution of outcomes of the computation does not change
significantly when one individual changes her input data. This is a very strong
privacy guarantee: anything significant about any individual that an
adversary could learn from the algorithm's output, he could also learn
were the individual not participating in the database at all---and this
holds true no matter what auxiliary information the adversary may have.
This definition guarantees privacy of an individual's sensitive data,
while allowing the computation to respond when a large number of
individuals change their data, as any useful computation must do.

\subsection{Our Results}

In this paper we initiate a systematic study of designing algorithms for
combinatorial optimization problems under the constraint of differential
privacy. Here is a short summary of some of the main contributions of
our work.

\begin{itemize}
\item While the exponential mechanism of \cite{MT07} is an easy way to obtain
  \emph{computationally inefficient} private approximation algorithms
  for some problems, the approximation guarantees given by a direct
  application of this can be far from optimal (e.g., see our results on
  min-cut and weighted set cover). In these cases, we have to use
  different techniques---often more sophisticated applications of the
  exponential mechanism---to get good (albeit computationally expensive)
  solutions.

\item However, we want our algorithms to be \emph{computationally
    efficient} and \emph{private} at the same time: here we cannot use
  the exponential mechanism directly, and hence we develop new
  algorithmic ideas.  We give private algorithms for a wide variety of
  \emph{search} problems, where we must not only approximate the
  \emph{value} of the solution, but also produce a solution that
  optimizes this value. See \lref[Table]{tab:results} for our results.

\item For some problems, unfortunately, just outputting an explicit
  solution might leak private information. For example, if we output a
  vertex cover of some graph explicitly, any pair of vertices not output
  reveals that they do not share an edge ---so any private explicit
  vertex cover algorithm must output $n-1$ vertices. To overcome this
  hurdle, we instead privately output an implicit representation of a
  small vertex cover--- we view vertex cover as a location problem, and
  output an orientation of the edges.  Each edge can cover itself using the end point that it points to. The orientation is output privately, and the resulting
  vertex cover approximates the optimal vertex cover well. We deal with
  similar representational issues for other problems like set cover as
  well.

\item We also show lower bounds on the approximation guarantees
  regardless of computational considerations.  For example, for vertex
  cover, we show that any $\epsilon$-differentially private algorithm
  must have an approximation guarantee of $\Omega(1/\epsilon)$.  We show
  that each of our lower bounds are tight: we give (computationally
  inefficient) algorithms with matching approximation guarantees.

\item Our results have implications beyond privacy as well:
  Papadimitriou et al.\ \cite{PSS08} introduce the {\em Combinatorial
    Public Project} problem, a special case of submodular
  maximization, and show that the problem can be well approximated by
  either a truthful mechanism or an efficient algorithm, but not by
  both simultaneously. In contrast to this negative result, we show
  that under differential privacy (which can be interpreted as an
  approximate but robust alternative to truthfulness) we can achieve
  the same approximation factor as the best non-truthful algorithm,
  plus an additive logarithmic loss.

\item Finally, we develop a private amplification lemma: we show how to take
  private algorithms that gives bounds in expectation and efficiently
  convert them (privately) into bounds with high probability. This
  answers an open question in the paper of Feldman et al.~\cite{FFKN09}.
\end{itemize}

\begin{table*} 
\begin{minipage}[center]{\textwidth}
\begin{tabular}{|r|c|cc|c|}
\hline
& Non-private & Efficient Algorithms & & Information Theoretic\\ \hline
Vertex Cover & $2 \times \OPT$ \cite{Pitt85} & $(2 + 16/\epsilon) \times \OPT$ & & $\Theta(1/\epsilon) \times \OPT$ \\ \hline
Wtd. Vertex Cover & $2 \times \OPT$  \cite{Hoc82} & $(16 + 16/\epsilon) \times \OPT$ & & $\Theta(1/\epsilon) \times \OPT$ \\ \hline
Set Cover  & $ \ln n \times \OPT$ \cite{Joh74} & $O(\ln n + \ln m/\epsilon) \times \OPT$ &$\dagger$& $\Theta(\ln m / \epsilon) \times \OPT$\\ \hline
Wtd. Set Cover & $\ln n \times \OPT$ \cite{Chv79} & $O(\ln n (\ln m + \ln\ln n)/\epsilon) \times \OPT$ &$\dagger$& $\Theta(\ln m /\epsilon) \times \OPT$ \\ \hline
Min Cut & $\OPT$ \cite{FF56} & $\OPT + O(\ln n/\epsilon)$ &$\dagger$& $\OPT + \Theta(\ln n /\epsilon)$ \\ \hline
CPPP  & $(1-1/e) \times \OPT$ \cite{NWF78} & $(1 - 1/e) \times \OPT - O(k \ln m/\epsilon)$ &$\dagger$& $\OPT - \Theta(k \ln (m/k) / \epsilon)$\\ \hline
$k$-Median & $(3 + \gamma) \times \OPT$ \cite{AGKMMP01} & $6 \times \OPT + O(k^2 \ln^2 n /\epsilon)$ &&  $\OPT + \Theta( k\ln (n/k) /\epsilon)$\footnote{\cite{FFKN09} independently prove a similar lower bound.} \\ \hline
\end{tabular}
\caption{\label{tab:results} Summary of Results. Results in the
  second and third columns are from this paper.}
\end{minipage}
\end{table*}
\lref[Table]{tab:results} summarizes the bounds we prove in this
paper.  For each problem, it reports (in the first column) the best
known non-private approximation guarantees, (in the second column) our
best efficient $\epsilon$-differentially private algorithms, and in
each (in the third column) case matching upper and lower bounds for
inefficient $\epsilon$-differentially private algorithms. For a few of
the efficient algorithms (marked with a $\dagger$) the guarantees are
only for an approximate form of differential privacy, incorporating a
failure probability $\delta$, and scaling the effective value of
$\epsilon$ up by $\ln (1/\delta)$.

\subsection{Related Work}
\label{sec:related-work}

Differential privacy is a relatively recent privacy definition (e.g.,
see
~\cite{DMNS06,Dwork-icalp,NRS07,BLR08,KLNRS,FFKN09,DNRRV09},
and see \cite{Dwo08} for an excellent survey), that tries to capture
the intuition of individual
privacy. 
Many algorithms in this framework have focused on measurement,
statistics, and learning tasks applied to statistical data sets,
rather than on processing and producing combinatorial objects.  One
exception to this is the Exponential Mechanism of \cite{MT07}
which allows the selection from a set of discrete alternatives.

Independently, Feldman et al. \cite{FFKN09} also consider the problem of privately approximating $k$-medians for points in $\Re^d$. Their model differs slightly from ours, which makes the results largely incomparable: while our results for general metrics translated to $\Re^d$ give smaller additive errors than theirs, we only output a $k$-median approximation whereas they output coresets for the problem. Their lower bound argument for private coresets is similar to ours.

Prior work on Secure Function Evaluation (SFE) tells us that in fact
the minimum cut in a graph can be computed in a distributed fashion in
such a way that computations {\em reveals nothing that cannot be
  learnt from the output of the computation}. While this is a strong
form of a privacy guarantee, it may be unsatisfying to an individual
whose private data can be inferred from the privately computed
output. Indeed, it is not hard to come up with instances where an
attacker with some limited auxiliary information can infer the
presence or absence of specific edges from local information about the
minimum cut in the graph. By relaxing the whole input privacy
requirement of SFE, differential privacy is able to provide
unconditional per element privacy, which SFE need not provide if the
output itself discloses properties of input.

Feigenbaum et al.~\cite{FIMNSW} extend the notion of SFE to NP hard problems for which efficient algorithms must output an approximation to the optimum, unless P=NP. They defined as {\em functional privacy} the constraint that two inputs with the same output value
(e.g. the size of an optimal vertex cover) must produce the same value
under the approximation algorithm.
Under this constraint, Halevi et
al.~\cite{HKKN01} show that approximating the value of vertex cover to
within $n^{1-\xi}$ is as hard as computing the value itself, for any
constant $\xi$. These hardness results were extended to {\em search}
problems by Beimel et al.~\cite{BCNW06}, where the constraint is
relaxed to only equate those inputs whose sets of optimal solutions
are identical. These results were extended and strengthened by Beimel
et al.~\cite{BHN07, BMNW07}.

Nonetheless, Feigenbaum et al.~\cite{FIMNSW} and others show a number
of positive approximation results under versions of the functional
privacy model. Halevi et al.~\cite{HKKN01} provide positive results in
the function privacy setting when the algorithm is permitted to leak
few bits (each equivalence class of input need not produce identical
output, but must be one of at most $2^b$ possible outcomes).  Indyk
and Woodruff also give some positive results for the approximation of
$\ell_2$ distance and a nearest neighbor problem \cite{IW06}. However,
as functional privacy extends SFE, it does not protect sensitive data that can be inferred from the output.

Nevertheless, SFE provides an implementation of any function in a distributed setting such that nothing other than the output of the function is revealed. One can therefore run a differentially private algorithm is a distributed manner using SFE (see e.g.~\cite{DKMMN06,BeimelNO08}), in the absence of a trusted curator.

\section{Definitions}
\label{sec:notation}

Differential privacy is a privacy definition for computations run
against sensitive input data sets. Its requirement, informally, is that
the computation behaves nearly identically on two input data sets that
are nearly identical; the probability of any outcome must not increase
by more than a small constant factor when the input set is altered by a
single element. Formally,

\begin{definition}[\cite{DMNS06}]
  We say a randomized computation $M$ has $\epsilon$-{\em differential
    privacy} if for any two input sets $A$ and $B$ with symmetric
  difference one, and for any set of outcomes $S \subseteq Range(M)$,
  \vspace{-0.05in}
  \begin{eqnarray}
    \Pr[M(A) \in S] & \le & \exp(\epsilon) \times \Pr[M(B) \in S] \; .
  \end{eqnarray}
\end{definition}
\vspace{-0.05in} The definition has several appealing properties from a
privacy perspective. One that is most important for us is that arbitrary
sequences of differentially private computations are also differentially
private, with an $\epsilon$ parameter equal to the sum of those
comprising the sequence. This is true even when subsequent computations
can depend on and incorporate the results of prior differentially
private computations~\cite{DKMMN06}, allowing repetition of
differentially private steps to improve solutions.

\subsection{Approximate Differential Privacy}

One relaxation of differential privacy \cite{DKMMN06} allows a small additive term in the bound:
\begin{definition}
  We say a randomized computation $M$ has $\delta$-{\em approximate}
  $\epsilon$-{\em differential privacy} if for any two input sets $A$
  and $B$ with symmetric difference one, and for any set of outcomes $S
  \subseteq Range(M)$,
\ifnum\fullversion=1
  \begin{eqnarray}
\Pr[M(A) \in S] & \le & \exp(\epsilon) \times \Pr[M(B) \in S] + \delta\; .
  \end{eqnarray}
\else
  \begin{eqnarray}
\;\;\;\;\;\;\;\;\Pr[M(A) \in S] & \le & \exp(\epsilon) \times \Pr[M(B) \in S] + \delta\; .
  \end{eqnarray}
\fi
\end{definition}

The flavor of guarantee is that although not all events have their
probabilities preserved, the alteration is only for very low probability
events, and is very unlikely to happen. The $\delta$ is best thought of
as $1/poly(n)$ for a data set containing some subset of $n$ candidate
records. We note that there are stronger notions of approximate differential privacy (c.f.~\cite{MachanavajjhalaKAGV08}), but in our settings, they are equivalent upto $poly(n)$ changes in $\delta$. We therefore restrict ourselves to this definition here.

\subsection{The Exponential Mechanism}

One particularly general tool that we will often use is the exponential mechanism of \cite{MT07}. This construction allows differentially private computation over arbitrary domains and ranges, parametrized by a query function $q(A,r)$ mapping a pair of input data set $A$ (a multiset over some domain) and candidate result $r$ to a real valued ``score''. With $q$ and a target privacy value $\epsilon$, the mechanism selects an output with exponential bias in favor of high scoring outputs:
\vspace{-0.05in}
\begin{eqnarray}
Pr[{\mathcal E}_q^\epsilon(A) = r] & \propto & \exp(\epsilon q(A,r)) \; .
\end{eqnarray}

If the query function $q$ has the property that any two adjacent data
sets have score within $\Delta$ of each other, for all possible outputs
$r$, the mechanism provides $2\epsilon\Delta$-differential privacy.
Typically, we would normalize $q$ so that $\Delta = 1$.
We will be using this mechanism almost exclusively over discrete ranges, where we can derive the following simple analogue of a theorem of \cite{MT07}, that the probability of a highly suboptimal output is exponentially low:
\begin{theorem}
  The exponential mechanism, when used to select an output $r \in R$ gives $2\epsilon\Delta$-differential privacy, letting $R_\OPT$ be the subset of $R$ achieving $q(A,r) = \max_r q(A,r)$, ensures that
\vspace{-0.05in}
\ifnum\fullversion=1
  \begin{eqnarray}
    \label{eqn:expmech}
    \Pr[q(A, \mathcal{E}_q^\epsilon(A)) < \max_r q(A,r) -
    \ln(|R|/|R_\OPT|)/\epsilon - t/\epsilon] & \le & \exp(-t) \; .
  \end{eqnarray}
\else
   \begin{equation*}
        \Pr[q(A, \mathcal{E}_q^\epsilon(A)) < \max_r q(A,r) -
    \ln(|R|/|R_\OPT|)/\epsilon - t/\epsilon]
    \end{equation*}
    \vspace{-0.3in}
    \begin{equation}
    \label{eqn:expmech}
    \le
     \exp(-t) \; .
    \end{equation}
\fi
\end{theorem}

The proof of the theorem is almost immediate: any outcome with score less
than $\max_r q(A,r) - \ln(|R|/|R_\OPT|)/\epsilon - t/\epsilon$ will have normalized probability at most $\exp(-t)/|R|$; each has weight at most $\exp(\OPT - t)|R_\OPT|/|R|$, but is normalized by at least $|R_\OPT|\exp(\OPT)$ from the optimal outputs. As there are at most $|R|$ such
outputs their cumulative probability is at most~$\exp(-t)$.

\section{Private Min-Cut}

Given a graph $G = (V,E)$ the minimum cut problem is to find a cut
$(S,S^c)$ so as to minimize $E(S,S^c)$. In absence of privacy
constraints, this problem is efficiently solvable exactly. However,
outputting an exact solution violates privacy, as we show in Section
\ref{sec:mincut-lb}. Thus, we give an
algorithm to output a cut within additive $O(\log n/\epsilon)$ edges
of optimal.

The algorithm has two
stages: First, given a graph $G$, we add edges to the graph to raise the
cost of the min cut to at least $4\ln n/\epsilon$, in a differentially
private manner. Second, we deploy the exponential mechanism over all
cuts in the graph, using a theorem of Karger to show that for graphs
with min cut at least $4\ln n/\epsilon$ the number of cuts within
additive $t$ of \OPT\ increases no faster than exponentially with $t$.
Although the exponential mechanism takes time exponential in $n$, we can
construct a polynomial time version by considering only the polynomially
many cuts within $O(\ln n/\epsilon)$ of \OPT. Below, let $Cost(H,(S,S^c))$ denote the size $E_H(S,S^c)$ of the cut $(S,S^c)$ in a graph $H$.

\begin{algorithm}
  \caption{The Min-Cut Algorithm}\label{alg:mincut}
  \begin{algorithmic}[1]
    \STATE \textbf{Input:} $G = (V,E)$,$\epsilon$.
    \STATE \textbf{Let} $H_0 \subset H_1,  \ldots, \subset H_{n\choose 2}$ be arbitrary strictly increasing sets of edges on $V$.
    \STATE \textbf{Choose} index $i\in[0,{n\choose 2}]$ with probability proportional to $\exp(-\epsilon |\OPT(G \cup H_i) - 8\ln n/\epsilon|)$.
    \STATE \textbf{Choose} a subset $S \in 2^V \setminus \{\emptyset,V\}$ with probability proportional to $\exp(-\epsilon Cost(G \cup H_i, (S,S^c)))$.
    \STATE \textbf{Output} the cut $C=(S,S^c)$.
\end{algorithmic}
\end{algorithm}

Our result relies on a result of Karger about the number of near-minimum
cuts in a graph~\cite{Karger-cuts}
\begin{lemma}[\cite{Karger-cuts}]
  \label{KargerLemma}
  For any graph $G$ with min cut $C$, there are at most $n^{2\alpha}$
  cuts of size at most $\alpha C$.
\end{lemma}

By enlarging the size of the min cut in $G \cup H_i$ to at least $ 4\ln
n /\epsilon$, we ensure that the number of cuts of value $\OPT(G \cup
H_i) + t$ is bounded by $n^2\exp(\epsilon t/2)$. The downweighting of
the exponential mechanism will be able to counteract this growth in
number and ensure that we select a good cut.

\begin{theorem}
For any graph $G$, the expected cost of ALG is at most $\OPT + O(\ln n/\epsilon)$.
\end{theorem}
\begin{proof}
  First, we argue that the selected index $i$ satisfies $4\ln
  n/\epsilon < \OPT(G \cup H_i) < \OPT(G) + 12\ln n/\epsilon$ with
  probability at least $1-1/n^2$. For $\OPT > 8\ln n /\epsilon$,
  Equation \ref{eqn:expmech} ensures that the probability of exceeding
  the optimal choice ($H_0$) by $4 \ln n /\epsilon$ is at most $1-
  1/n^2$. Likewise, for $\OPT < 8\ln n/\epsilon$, there is some
  optimal $H_i$ achieving min cut size $8\ln n/\epsilon$, and the
  probability we end up farther away than $4\ln n/\epsilon$ is at most
  $1-1/n^2$.

  Assuming now that $\OPT(G\cup H_i) > 4\ln n/\epsilon$, Karger's
  lemma argues that the number $c_t$ of cuts in $G \cup H_i$ of cost
  at most $\OPT(G \cup H_i) + t$ is at most $n^2 \exp(\epsilon t
  /2)$. As we are assured a cut of size $\OPT(G\cup H_i)$ exists, each
  cut of size $\OPT(G\cup H_i) + t$ will receive probability at most
  $\exp(-\epsilon t)$. Put together, the probability of a cut
  exceeding $\OPT(G\cup H_i) + b$ is at most
\ifnum\fullversion=1
\begin{eqnarray}
  \Pr[Cost(G\cup H_i, C) > \OPT(G \cup H_i) + b] & \le & \sum_{t > b} \exp(-\epsilon t) (c_t - c_{t-1}) \\
  & \le & (\exp(\epsilon)-1)\sum_{t > b} \exp(-\epsilon t) c_t \\
  & \le & (\exp(\epsilon)-1)\sum_{t > b} \exp(-\epsilon t/2) n^2 \end{eqnarray}
\else
\begin{align*}
&\Pr[Cost(G\cup H_i, C) > \OPT(G \cup H_i) + b] \\
  & \le  \sum_{t > b} \exp(-\epsilon t) (c_t - c_{t-1})\\
  & \le  (\exp(\epsilon)-1)\sum_{t > b} \exp(-\epsilon t) c_t \\
  & \le  (\exp(\epsilon)-1)\sum_{t > b} \exp(-\epsilon t/2) n^2 \end{align*}
\fi
The sum telescopes to $\exp(-\epsilon b/2) n^2 / (\exp(\epsilon/2) -
1)$, and the denominator is within a constant factor of the leading
factor of $(\exp(\epsilon) - 1)$, for $\epsilon<1$. For $b = 8\ln n /
\epsilon$, this probability becomes at most $1/n^2$.
\end{proof}

\begin{theorem}
The algorithm above preserves $2\epsilon$-differential privacy.
\end{theorem}
Note that the first instance of the exponential mechanism in our algorithm runs efficiently (since it is selecting from only ${n\choose 2}$ objects), but the second instance does not. We now describe how to achieve $(\eps,\delta)$-differential privacy efficiently.

First recall that using Karger's algorithm we can efficiently (with high probability) generate all cuts of size at most $k\OPT$ for any constant $k$. Indeed it is shown in~\cite{Karger-cuts} that in a single run of his algorithm, any such cut is output with probability at least $n^{-2k}$ so that $n^{2k+1}$ runs of the algorithm will output all such cuts except with an exponentially small probability.

Our efficient algorithm works as follows: in step 4 of Algorithm~\ref{alg:mincut}, instead of sampling amongst all possible cuts, we restrict attention to the set of cuts generated in $n^7$ runs of Karger's algorithm. We claim that the output distribution of this algorithm has statistical distance $O(1/n^2)$ from that of Algorithm~\ref{alg:mincut}, which would imply that we get $(\eps,O(\frac{1}{n^2}))$-differential privacy.

Consider a hypothetical algorithm that generates the cut $(S,S^c)$ as in Algorithm~\ref{alg:mincut} but then outputs {\sc {FAIL}} whenever this cut is not in the set of cuts generated by $n^7$ runs of Karger's. We first show that the probability that this algorithm outputs {\sc {FAIL}} is $O(\frac{1}{n^2})$. As shown above, $\OPT(G\cup H_i)$ is at least $4\ln n/\eps$ except with probability $\frac{1}{n^2}$. Conditioned on this, the cut chosen in Step 4 has cost at most $3\OPT(G \cup H_i)$ except with probability $\frac{1}{n^2}$. Since each such cut is in the sample except with exponentially small probability, the claim follows. Finally, note that this hypothetical algorithm can be naturally coupled with both the algorithms so that the outputs agree whenever the former doesn't output {\sc {FAIL}}. This implies the claimed bound on the statistical distance.
We remark that we have not attempted to optimize the running time here; both the running time and the value of $\delta$ can be improved by choosing a larger constant (instead of 8) in Step 3, at the cost of increasing the additive error by an additional constant.

\subsection{Lower Bounds}\label{sec:mincut-lb}

We next show that this additive error is unavoidable for any differentially private algorithm. The lower bound is information-theoretic and thus applies also to computationally inefficient algorithms.
\begin{theorem}
Any $\epsilon$-differentially private algorithm for min-cut must incur an expected additive $\Omega(\ln n /\epsilon)$ cost over \OPT, for any $\epsilon\in (3\ln n/n,\frac{1}{12})$.
\end{theorem}
\begin{proof}
Consider a $\ln n/3\epsilon$-regular graph $G=(V,E)$ on $n$ vertices such that the minimum cuts are exactly those that isolate a single vertex, and any other cut has size at least $(\ln n / 2\epsilon)$ (a simple probabilistic argument establishes the existence of such a $G$; in fact a randomly chosen $\ln n/3\epsilon$-regular graph has this property with high probability).

Let $M$ be an $\epsilon$-differentially private algorithm for the min-cut. Given the graph $G$, $M$ outputs a partition of $V$. Since there are $n=|V|$ singleton cuts, there exists a vertex $v$ such that the mechanism $M$ run on $G$ outputs the cut $(\{v\},V\setminus \{v\})$ with probability at most $1/n$, i.e. $$Pr[M(V,E) = (\{v\},V\setminus \{v\}) \leq \frac{1}{n}.$$

Now consider the graph $G'=(V,E')$, with the edges incident on $v$ removed from $G$, i.e. $E'=E\setminus \{e: v\in e\}$. Since $M$ satisfies $\epsilon$-differential privacy and $E$ and $E'$ differ in at most $\ln n / 3\epsilon$ edges, $$Pr[M(V,E') = (\{v\},V\setminus \{v\})] \leq 1/n^{1/3}.$$

Thus with probability $(1-\frac{1}{n^{\frac{1}{3}}})$, $M(G')$ outputs a cut other than the minimum cut $(\{v\},V\setminus\{v\})$.  But all other cuts, even with these edges removed, cost at least $(\ln n/6\epsilon)$. Since \OPT is zero for $G'$, the claim follows.
\end{proof}

\section{Private $k$-Median}
\label{sec:k-median}

We next consider a private version of the metric $k$-median problem: There
is a pre-specified set of points $V$ and a metric on them, $d:V\times
V\rightarrow \Real$. There is a (private) set of demand points $D
\subseteq V$. We wish to select a set of medians $F \subset V$ with $|F|
= k$ to minimize the quantity $\cost(F) = \sum_{v \in D}d(v,F)$ where
$d(v,F) = \min_{f\in F}d(v,f)$. Let $\Delta = \max_{u,v \in V}d(u,v)$ be
the diameter of the space.

As we show in Section \ref{sec:kmedian-lb}, any privacy-preserving
algorithm for $k$-median must incur an additive loss of $\Omega(\Delta\cdot k\ln
  (n/k)/\epsilon)$, regardless of computational constraints.
We observe that running the exponential mechanism to choose one
of the $n\choose k$ subsets of medians gives an (computationally
inefficient) additive guarantee.
\begin{theorem}
  Using the exponential mechanism to pick a set of $k$ facilities
  gives an $O({n\choose k}poly(n))$-time $\epsilon$-differentially
  private algorithm that outputs a solution with expected cost $\OPT +
  O(k\Delta\log n/\epsilon)$.
\end{theorem}

We next give a polynomial-time algorithm that gives a slightly worse
approximation guarantee. Our algorithm is based on the local search
algorithm of Arya {\em et al.}~\cite{AGKMMP01}. We start with an
arbitrary set of $k$ medians, and use the exponential mechanism to
look for a (usually) improving swap. After running this local search
for a suitable number of steps, we select a good solution from amongst
the ones seen during the local search. The following result shows that
if the current solution is far from optimal, then one can find
improving swaps.

\begin{theorem}[Arya et al.~\cite{AGKMMP01}]
  For any set $F \subseteq V$ with $|F| = k$, there exists a set of $k$
  swaps $(x_1,y_1),\ldots,(x_k,y_k)$ such that
  $\sum_{i=1}^k(\cost(F)-\cost(F - \{x_i\} + \{y_i\}))
  \geq \cost(F) - 5\OPT$.

\end{theorem}
\begin{corollary}
  \label{swapCorollary}
  For any set $F \subseteq V$ with $|F| = k$, there exists some swap
  $(x,y)$ such that
  $$\cost(F)-\cost(F - \{x_i\} + \{y_i\}) \geq
  \frac{\cost(F) - 5\OPT}{k}.$$
\end{corollary}

\begin{algorithm}
  \caption{The $k$-Median Algorithm}
  \begin{algorithmic}[1]
    \STATE \textbf{Input:} $V$, Demand points $D \subseteq V$, $k$,$\epsilon$.
    \STATE \textbf{let} $F_1 \subset V$ arbitrarily with $|F_1| = k$, $\epsilon' \leftarrow \epsilon/(2\Delta(T+1))$.
    \FOR{$i=1$ to $T$}
        \STATE Select $(x,y) \in F_i \times (V\setminus F_i)$ with probability proportional to $\exp(-\epsilon' \times \cost(F_i - \{x\} + \{y\}))$.
        \STATE \textbf{let} $F_{i+1} \leftarrow F_i - \{x\} + \{y\}$.
    \ENDFOR
    \STATE Select $j$ from $\{1, 2, \ldots, T\}$ with probability proportional to $\exp(-\epsilon' \times \cost(F_j))$.
    \STATE \textbf{output} $F_j$.
\end{algorithmic}
\end{algorithm}

\begin{theorem}
  \label{thm:kmedian}
  Setting $T = 6k \ln n$ and $\epsilon' = \epsilon/(2\Delta
  (T+1))$, the $k$-median algorithm provides $\epsilon$-differential
  privacy and except with probability $O(1/\text{poly}(n))$ outputs a
  solution of cost at most $6\OPT + O(\Delta k^2\log^2 n/\epsilon)$.
\end{theorem}

\begin{proof}
  We first prove the privacy. Since the $\cost$ function has
  sensitivity $\Delta$, Step~4 of the algorithm preserves
  $2\epsilon'\Delta$ differential privacy. Since Step~4 is run at most
  $T$ times and privacy composes additively, outputting all of the $T$
  candidate solutions would give us $(2\epsilon'\Delta T)$
  differential privacy. Picking out a good solution from the $T$
  candidates costs us another $2\epsilon'\Delta$, leading to the
  stated privacy guarantee.

  We next show the approximation guarantee.  By Corollary
  \ref{swapCorollary}, so long as $\cost(F_i) \geq 6\OPT$, there
  exists a swap $(x,y)$ that reduces the cost by at least
  $\cost(F_i)/6k$. As there are only $n^2$ possible swaps, the
  exponential mechanism ensures through \eqref{eqn:expmech} that we
  are within additive $4 \ln n /\epsilon'$ with probability at least
  $1-1/n^2$. When $\cost(F_i) \geq 6\OPT + 24 k\ln n/\epsilon'$, with
  probability $1-1/n^2$ we have $\cost(F_{i+1}) \le (1-1/6k) \times
  \cost(F_i)$.

  This multiplicative decrease by $(1-1/6k)$ applies for as long as
  $\cost(F_i) \geq 6\OPT + 24 k\ln n/\epsilon'$. Since $\cost(F_0)
  \leq n\Delta$, and $n\Delta(1-1/6k)^T \leq \Delta \leq 24 k\ln
  n/\epsilon'$, there must exist an $i<T$ such that $\cost(F_i) \leq
  6\OPT + 24 k\ln n/\epsilon'$, with probability at least $(1-T/n^2)$.

Finally, by applying the exponential mechanism again in the final
stage, we select from the $F_i$ scoring within an additive $4\ln
n/\epsilon'$ of the optimal visited $F_i$ with probability at least
$1-1/n^2$, again by \eqref{eqn:expmech}. Plugging in the value of
$\epsilon'$, we get the desired result. Increasing the constants in
the additive term can drive the probability of failure to an
arbitrarily small polynomial.
\end{proof}

\subsection{$k$-Median Lower Bound}\label{sec:kmedian-lb}

\begin{theorem}
  \label{thm:kmed-lbd}
  Any $\epsilon$-differentially private algorithm for the $k$-median
  problem must incur cost $\OPT + \Omega(\Delta\cdot k\ln
  (n/k)/\epsilon)$ on some inputs.
\end{theorem}
\begin{proof}
  Consider a point set $V=[n]\times[L]$ of $nL$ points, with $L =
  \ln(n/k)/10\epsilon$, and a distance function $d((i,j),(i',j')) =
  \Delta$ whenever $i\neq i'$ and $d((i,j),(i,j'))=0$. Let $M$ be a
  differentially private algorithm that takes a subset $D \subseteq V$
  and outputs a set of $k$ locations, for some $k < \frac{n}{4}$.  Given
  the nature of the metric space, we assume that $M$ outputs a
  $k$-subset of $[n]$. For a set $A \subseteq [n]$, let $D_A = A\times
  [L]$. Let $A$ be a size-$k$ subset of $V$ chosen at random.

  We claim that that $\E_{A, M}[|M(D_A) \cap A|] \leq \frac{k}{2}$ for
  any $\epsilon$-differentially private algorithm $M$. Before we prove
  this claim, note that it implies the expected cost of $M(D_A)$ is
  $\frac{k}{2} \times \Delta L$, which proves the claim since $\OPT = 0$.

  Now to prove the claim: define $\phi := \frac{1}{k} \E_{A,M}[|A \cap
  M(D_A)|]$.  We can rewrite
  \ifnum\fullversion=1
  \begin{gather*}
    k\cdot \phi = \E_{A,M}[|A \cap M(D_A)|] = k \cdot \E_{i\in
      [n]}\E_{A\setminus\{i\},M}[\mathbf{1}_{i \in M(D_A)}]
  \end{gather*}
  \else
  \begin{eqnarray*}
    k\cdot \phi &=& \E_{A,M}[|A \cap M(D_A)|] \\&=& k \cdot \E_{i\in
      [n]}\E_{A\setminus\{i\},M}[\mathbf{1}_{i \in M(D_A)}]
  \end{eqnarray*}
  \fi
  Now changing $A$ to $A' := A \setminus \{i\} + \{i'\}$ for some random
  $i'$ requires altering at most $2L$ elements in $D_{A'}$, which by the
  differential privacy guarantee should change the probability of the
  output by at most $e^{2\epsilon L} = (n/k)^{1/5}$. Hence
  \begin{gather*}
    \E_{i\in [n]}\E_{A',M}[\mathbf{1}_{i \in M(D_{A'})}] \geq \phi \cdot
    (k/n)^{1/5}.
  \end{gather*}
  But the expression on the left is just $k/n$, since there at at most
  $k$ medians. Hence $\phi \leq (k/n)^{4/5} \leq 1/2$, which proves the
  claim.
\end{proof}

\begin{corollary}
  \label{cor:facloc-lbd}
  Any $1$-differentially private algorithm for uniform facility location
  that outputs the set of chosen facilities must have approximation
  ratio $\Omega(\sqrt{n})$.
\end{corollary}
\begin{proof}
  We consider instances defined on the uniform metric on $n$ points,
  with $d(u,v)=1$ for all $u,v$, and facility opening cost
  $f=\frac{1}{\sqrt{n}}$. Consider a $1$-differentially private
  mechanism $M$ when run on a randomly chosen subset $A$ of size $k =
  \sqrt{n}$. Since $\OPT$ is $kf=1$ for these instances, any
  $o(\sqrt{n})$-approximation must select at least $\frac{k}{2}$
  locations from $A$ in expectation. By an argument analogous to the
  above theorem, it follows that any differentially private $M$ must
  output $n/20$ of the locations in expectation. This leads to a
  facility opening cost of $\Omega(\sqrt{n})$.
\end{proof}

\subsection{Euclidean Setting}
Feldman {\em et al.}~\cite{FFKN09} study private coresets for the $k$-median problem when the input points are in $\Re^d$. For $P$ points in the unit ball in $\Re^d$, they give coresets with $(1+\eps)$ multiplicative error, and additive errors about $O(k^2 d^2 \log^2 P)$  and $O(16kd)^{2d} d^{3/2} \log P \log dk)$ respectively for their inefficient and efficient algorithms. Since Euclidean $k$-median has a PTAS, this leads to $k$-median approximations with the same guarantees. We can translate our results to their setting by looking at a $(1/P)$-net of the unit ball as the candidate set of $n$-points, of which some may appear. This would lead to an inefficient algorithm with additive error $O(kd\log P)$, and an efficient algorithm with additive error $O(k^2d^2\log^2 P)$. The latter has a multiplicative error of 6 and hence our efficient algorithms are incomparable. Note that coresets are more general objects than just the $k$-median solution.

\section{Vertex Cover}
\label{sec:VC}

We now turn to the problem of (unweighted) vertex cover, where we want
to pick a set $S$ of vertices of minimal size so that every edge in the
graph is incident to at least one vertex in $S$. In the
privacy-preserving version of the problem, the private information we
wish to conceal is the presence of absence of each edge.

\medskip\noindent\emph{Approximating the Vertex Cover \underline{Size}.}
As mentioned earlier, even approximating the vertex cover size was shown
to be polynomially inapproximable under the constraint of
\emph{functional} privacy~\cite{HKKN01, BCNW06}. On the other hand, it
is easy to approximate the size of the optimal vertex cover under
differential privacy: twice the size of a maximum matching is a
2-approximation to the optimal vertex cover, and this value only changes
by at most two with the presence or absence of a single edge. Hence,
this value plus $\text{Laplace}(2/\epsilon)$ noise provides
$\epsilon$-differential privacy~\cite{DMNS06}. (Here it is important
that we use \emph{maximum} rather than just maximal matchings, since the
size of the latter is not uniquely determined by the graph, and the
presence or absence of an edge may dramatically alter the size of the
solution.) Interestingly enough, for \emph{weighted} vertex cover with
maximum weight $w_{\max}$ (which we study in
\lref[Section]{sec:weighted-vc}), we have to add in
$\text{Lap}(w_{\max}/\epsilon)$ noise to privately estimate the weight
of the optimal solution, which can be much larger than $\OPT$ itself.
The mechanism in \lref[Section]{sec:weighted-vc} avoids this barrier by
outputting an implicit representation of the vertex cover, and hence
gives us a $O(1/\epsilon)$ multiplicative approximation with
$\epsilon$-differential privacy.

\medskip\noindent\emph{The Vertex Cover \underline{Search} Problem.}  If
we want to find a vertex cover (and not just estimate its size), how can
we do this privately? In covering problems, the (private) data imposes
hard constraints on the a solution, making them quite different from,
say, min-cut. Indeed, while the private data only influences the {\em
  objective function} in the min-cut problem, the data determines the
{\em constraints} defining feasible solutions in the case of the vertex
cover problem. This hard covering constraint make it impossible to
actually output a small vertex cover privately: as noted in the
introduction, any differentially private algorithm for vertex cover that
outputs an explicit vertex cover (a subset of the $n$ vertices) must
output a cover of size at least $n-1$ with probability 1 on any input,
an essentially useless result.

In order to address this challenge, we require our algorithms to output
an \emph{implicit representation} of a cover: we privately output an orientation of the edges. Now for each edge, if we pick the endpoint that it points to, we
clearly get a vertex cover. Our analysis ensures that this vertex cover has
size not much larger than the size of the optimal vertex cover for the
instance. Hence, such an orientation may be viewed as a
privacy-preserving set of instructions that allows for the construction
of a good vertex cover in a distributed manner: in the case of the
undercover agents mentioned in the introduction, the complete set of
active dropoff sites (nodes) is not revealed to the agents, but an
orientation on the edges tells each agent which dropoff site to use, if
she is indeed an active agent. Our algorithms in fact output a permutation of all the vertices of the graph. Each edge can be considered oriented towards the endpoint appearing earlier in the permutation. Our lower bounds apply to the more general setting where we are allowed to output any orientation (and hence are stronger).

\subsection{The Algorithm for Unweighted Vertex Cover}

Our (randomized) algorithm will output a permutation, and the vertex cover will be defined by
picking, for each edge, whichever of its endpoints appears first in the
permutation. We show that this vertex cover will be $(2 +
O(1/\epsilon))$-approximate and $\epsilon$-differentially private. Our
algorithm is based on a simple (non-private) 2-approximation to vertex
cover~\cite{Pitt85} that repeatedly selects an uncovered edge uniformly
at random, and includes a random endpoint of the edge. We can view the
process, equivalently, as selecting a vertex at random with probability
proportional to its uncovered degree. We will take this formulation and
mix in a uniform distribution over the vertices, using a weight that
will grow as the number of remaining vertices decreases.

Let us start from $G_1 = G$, and let $G_i$ be the graph with $n - i + 1$
vertices remaining. We will write $d_v(G)$ for the degree of vertex $v$
in graph $G$.  The algorithm $ALG$ in step $i$ chooses from the $n - i +
1$ vertices of $G_i$ with probability proportional to $d_v(G_i) + w_i$,
for an appropriate sequence $\left<w_i\right>$.  Taking $w_i =
(4/\epsilon) \times (n/(n-i+1))^{1/2}$ provides $\epsilon$-differential
privacy and a $(2 + 16/\epsilon)$ approximation factor, the proof of
which will follow from the forthcoming \lref[Theorem]{thm:privacy} and
\lref[Theorem]{thm:accuracy}.

As stated the algorithm outputs a sequence of vertices, one per iteration. As remarked above, this permutation defines a vertex cover by picking the earlier occurring end point of each edge.

\begin{algorithm}
  \caption{Unweighted Vertex Cover}
  \begin{algorithmic}[1]
  \STATE \textbf{let} $n \leftarrow |V|$, $V_1 \leftarrow V,  E_1\leftarrow E$.
    \FOR{$i=1,2, \ldots, n$}
    \STATE \textbf{let} $w_i \leftarrow (4/\epsilon) \times \sqrt{n/(n-i+1)}$.
    \STATE \textbf{pick} a vertex $v \in V_i$ with probability proportional to $d_{E_i}(v)+w_i$.
    \STATE \textbf{output} $v$. \textbf{let} $V_{i+1} \leftarrow V_i \setminus \{v\}$, $E_{i+1} \leftarrow E_i \setminus (\{v\}\times V_i)$.
    \ENDFOR
\end{algorithmic}
\end{algorithm}

\begin{theorem}[Privacy]
  \label{thm:privacy}
  ALG satisfies $\epsilon$-differential privacy for the settings of $w_i$ above.
\end{theorem}

\begin{proof}
  For any two sets of edges $A$ and $B$, and any permutation $\pi$, let
  $d_i$ be the degree of the $i^{th}$ vertex in the permutation $\pi$
  and let $m_i$ be the remaining edges, both ignoring edges incident to
  the first $i-1$ vertices in $\pi$.
\ifnum\fullversion=1
  \begin{gather*}
    \frac{\Pr[ALG(A) = \pi]}{\Pr[ALG(B) = \pi]} = \prod_{i = 1}^n
    \frac{(w_i + d_i(A))/((n-i+1)w_i + 2m_i(A))}{(w_i + d_i(B))/((n-i+1)w_i +
      2m_i(B))} \; .
  \end{gather*}
\else
  \begin{gather*}
   \frac{\Pr[ALG(A) = \pi]}{\Pr[ALG(B) = \pi]}\;\;\;\;\;\;\;\;\;\;\;\;\;\;\;\;\;\;\;\;\;\;\;\;\;\;\;\;\;\;\;\;\;\;\;\;\;\;\;\;\;\;\;\;\;\;\;\;\\
    \;\;\; = \prod_{i = 1}^n
    \frac{(w_i + d_i(A))/((n-i+1)w_i + 2m_i(A))}{(w_i + d_i(B))/((n-i+1)w_i +
      2m_i(B))} \; .
  \end{gather*}
\fi
  When $A$ and $B$ differ in exactly one edge, $d_i(A) = d_i(B)$ for all
  $i$ except the first endpoint incident to the edge in the difference.
  Until this term $m_i(A)$ and $m_i(B)$ differ by exactly one, and after
  this term $m_i(A) = m_i(B)$. The number of nodes is always equal, of
  course. Letting $j$ be the index in $\pi$ of the first endpoint of the
  edge in difference, we can cancel all terms after $j$ and rewrite
\ifnum\fullversion=1
  \begin{gather*}
    \frac{\Pr[ALG(A) = \pi]}{\Pr[ALG(B) = \pi]} = \frac{w_j +
      d_j(A)}{w_j + d_j(B)} \times \prod_{i \leq j} \frac{(n-i+1)w_i +
      2m_i(B)}{(n-i+1)w_i + 2m_i(A)} \; .
  \end{gather*}
\else
  \begin{gather*}
    \frac{\Pr[ALG(A) = \pi]}{\Pr[ALG(B) = \pi]}\;\;\;\;\;\;\;\;\;\;\;\;\;\;\;\;\;\;\;\;\;\;\;\;\;\;\;\;\;\;\;\;\;\;\;\;\;\;\;\;\;\;\;\;\;\;\;\;\\
     = \frac{w_j +
      d_j(A)}{w_j + d_j(B)} \times \prod_{i \leq j} \frac{(n-i+1)w_i +
      2m_i(B)}{(n-i+1)w_i + 2m_i(A)} \; .
  \end{gather*}
\fi
  An edge may have arrived in $A$, in which case $m_i(A)= m_i(B)+1$ for
  all $i \le j$, and each term in the product is at most one; moreover,
  $d_j(A) = d_j(B) + 1$, and hence the leading term is at most $1 +
  1/w_j < \exp(1/w_1)$, which is bounded by $\exp(\epsilon/2)$.

  Alternately, an edge may have departed from $A$, in which case the
  lead term is no more than one, but each term in the product exceeds
  one and their product must now be bounded. Note that $m_i(A) + 1 =
  m_i(B)$ for all relevant $i$, and that by ignoring all other edges we
  only make the product larger. Simplifying, and using $1 + x \le
  \exp(x)$, we see
\ifnum\fullversion=1
  \begin{gather*}
    \prod_{i \le j} \frac{(n-i+1)w_i + 2m_i(B)}{(n-i+1)w_i + 2m_i(A)}
    \;\; \le \;\; \prod_{i \le j} \frac{(n-i+1)w_i + 2}{(n-i+1)w_i + 0}
    \;\; = \;\; \prod_{i \le j} \left(1 + \frac{2}{(n-i+1)w_i}\right)
    \;\; \le \;\;  \exp \left( \sum_{i \le j} \frac{2}{(n-i+1)w_i} \right)
    \; .
  \end{gather*}
\else
  \begin{eqnarray*}
  &&  \prod_{i \le j} \frac{(n-i+1)w_i + 2m_i(B)}{(n-i+1)w_i + 2m_i(A)}\\
  & &\le  \prod_{i \le j} \frac{(n-i+1)w_i + 2}{(n-i+1)w_i + 0}\\
  & &=  \prod_{i \le j} \left(1 + \frac{2}{(n-i+1)w_i}\right)\\
  & &\le   \exp \left( \sum_{i \le j} \frac{2}{(n-i+1)w_i} \right).
  \end{eqnarray*}
\fi
  The $w_i$ are chosen so that $\sum_{i} 2/(n-i+1)w_i = (\eps/\sqrt{n}) \sum_i 1/2\sqrt{i}$ is at most $\eps$.
\end{proof}

\begin{theorem}[Accuracy]
  \label{thm:accuracy}
  For all $G$, $\E[{ALG}(G)] \; \le \; (2 + 2\avg_{i\le n} w_i) \times
  |OPT(G)| \; \le \; (2 + 16/\epsilon)|\OPT(G)|$.
\end{theorem}
\begin{proof}
  Let $OPT(G)$ denote an arbitrary optimal solution to the vertex cover
  problem on $G$.  The proof is inductive, on the size $n$ of $G$. For
  $G$ with $|OPT(G)| > n/2$, the theorem holds. For $G$ with $|OPT(G)|
  \le n/2$, the expected cost of the algorithm is the probability that
  the chosen vertex $v$ is incident to an edge, plus the expected cost
  of $ALG(G \setminus v)$.
\ifnum\fullversion=1
  \begin{eqnarray*}\label{eq:four}
    \E[ALG(G)] & = & \Pr[v \mbox{ incident on edge} ] + \E_v[\E_{}[ALG(G
    \setminus v)]] \; .
  \end{eqnarray*}
\else
  \begin{eqnarray*}\label{eq:four}
    \E[ALG(G)] & = & \Pr[v \mbox{ incident on edge} ]\\ &&+ \E_v[\E_{}[ALG(G
    \setminus v)]] \; .
  \end{eqnarray*}
\fi
  We will bound the second term using the inductive hypothesis. To bound
  the first term, the probability that $v$ is chosen incident to an edge
  is at most $(2mw_n + 2m)/(nw_n + 2m)$, as there are at most $2m$
  vertices incident to edges. On the other hand, the probability that we
  pick a vertex in $OPT(G)$ is at least $(|OPT(G)| w_n + m)/(nw_n+2m)$.
  Since $|OPT(G)|$ is non-negative, we conclude that
\ifnum\fullversion=1
  \[\Pr[v \mbox{ incident on edge}] \leq (2+2w_n)(m/(nw_n+2m)) \leq
  (2+2w_n)\Pr[v \in OPT(G)]\]
\else
\begin{eqnarray*}\Pr[v \mbox{ incident on edge}] &\leq& (2+2w_n)(m/(nw_n+2m))\\
&\leq &  (2+2w_n)\Pr[v \in OPT(G)]
\end{eqnarray*}
\fi
  Since $\one[v \in OPT(G)] \leq |OPT(G)| - |OPT(G\setminus v)|$, and
  using the inductive hypothesis, we get
  \begin{eqnarray*}
    \E[ALG(G)] & \le & (2 + 2w_n) \times (|OPT(G)| - \E_v[|OPT(G
    \setminus v)|]) + (2 + 2\avg_{i < n}  w_i) \times \E_v[|OPT(G
    \setminus v)|] \\ & = & (2 + 2w_n) \times |OPT(G)| + (2\avg_{i < n}
    w_i - 2w_n) \times \E_v[|OPT(G \setminus v)|]
  \end{eqnarray*}
  The probability that $v$ is from an optimal vertex cover is at least
  $(|OPT(G)|w_i + m)/(nw_i + 2m)$, as mentioned above, and (using
  $(a+b)/(c+d) \ge \min \{a/c, b/d\}$) is at least $\min \{|OPT(G)|/n ,
  1/2\} = |OPT(G)|/n$, since $|OPT(G)| < n/2$ by assumption.  Thus
  $\E[|OPT(G \setminus v)|]$ is bounded above by $(1 - 1/n) \times
  |OPT(G)|$, giving
  \begin{eqnarray*}
    \E[ALG(G)] & \le & (2 + 2w_n) \times |OPT(G)| + (2 \avg_{i < n} w_i
    - 2w_n) \times (1 - 1/n) \times |OPT(G)|  \; .
  \end{eqnarray*}
  Simplification yields the claimed results, and instantiating $w_i$ completes the proof.
\end{proof}

\ifnum\fullversion=1
\paragraph{Hallucinated Edges.}
Here is a slightly different way to implement the intuition behind the
above algorithm: imagine adding $O(1/\epsilon)$ ``hallucinated'' edges
to each vertex (the other endpoints of these hallucinated edges being
fresh ``hallucinated'' vertices), and then sampling vertices without
replacement proportional to these altered degrees. However, once (say)
$n/2$ vertices have been sampled, output the remaining vertices in
random order. This view will be useful to keep in mind for the weighted
vertex cover proof.  (A formal analysis of this algorithm is in
\lref[Appendix]{sec:aaron-proof}.)
\else
\paragraph{Hallucinated Edges.}
Here is a slightly different way to implement the intuition behind the
above algorithm: imagine adding $O(1/\epsilon)$ ``hallucinated'' edges
to each vertex (the other endpoints of these hallucinated edges being
fresh ``hallucinated'' vertices), and then sampling vertices without
replacement proportional to these altered degrees. However, once (say)
$n/2$ vertices have been sampled, output the remaining vertices in
random order. This view will be useful to keep in mind for the weighted
vertex cover proof.  (A formal analysis of this algorithm appears in the full version.)
\fi

\subsection{Weighted Vertex Cover}
\label{sec:weighted-vc}

In the weighted vertex cover problem, each vertex $V$ is assigned a
weight $w(v)$, and the cost of any vertex cover is the sum of the
weights of the participating vertices. One can extend the unweighted
2-approximation that draws vertices at random with probability
proportional to their uncovered degree to a weighted 2-approximation by
drawing vertices with probability proportional to their uncovered degree
divided by their weight. The differentially private analog of this
algorithm essentially draws vertices with probability proportional
to $1/\epsilon$ plus their degree, all divided by the weight of the
vertex; the algorithm we present here is based on this idea.

Define the \emph{score} of a vertex to be $s(v) = 1/w(v)$.  Our
algorithm involves hallucinating edges: to each vertex, we add in
$1/\epsilon$ hallucinated edges, the other endpoints of which are
imaginary vertices, whose weight is considered to be~$\infty$ (and hence
has zero score).  The score of an edge $e = (u,v)$ is defined to be
$s(e) = s(u) + s(v)$; hence the score of a fake edge $f$ incident on $u$
is $s(f) = s(u)$, since its other (imaginary) endpoint has infinite
weight and zero score. We will draw edges with probability proportional
to their score, and then select an endpoint to output with
probability proportional to its score.  In addition, once a substantial
number of vertices of at least a particular weight have been output,
we will output the rest of those vertices.

Assume the minimum vertex weight is $1$ and the maximum is $2^J$. For
simplicity, we round the weight of each vertex up to a power of~$2$, at
a potential loss of a factor of two in the approximation. Define the
$j^{th}$ weight class $V_j$ to be the set of vertices of weight $2^j$.
In addition, we will assume that $|V_j| = |V_{j+1}|$ for all weight
classes. In order to achieve this, we hallucinate additional fake
vertices. We will never actually output a hallucinated vertex. Let $N_j$
denote $|V_j|$.

\begin{algorithm}
  \caption{Weighted Vertex Cover}
  \begin{algorithmic}[1]
    \WHILE{not all vertices have been output}
    \STATE \textbf{pick} an uncovered (real or hallucinated) edge $e = (u,v)$
    with probability proportional to $s(e)$.
    \STATE \textbf{output} endpoint $u \in e$ with probability proportional to
    $s(u)$.
    \WHILE{there exists some weight class $V_j$ such that the number of
      nodes of class $j$ or higher that we've output is at least
      $N_j/2 = |V_j|/2$}
    \STATE \textbf{pick} the smallest such value of $j$
    \STATE \textbf{output} (``dump'') all remaining vertices in $V_j$ in random
    order.
    \ENDWHILE
    \ENDWHILE
\end{algorithmic}
\end{algorithm}

We imagine the $i^{th}$ iteration of the outer loop of the algorithm as
happening at \emph{time} $i$; note that one vertex is output in Step~3,
whereas multiple vertices might be output in Step~6. Let
$\widetilde{n}_i$ be the sum of the scores of all real vertices not
output before time $i$, and $\widetilde{m}_i$ be the sum of the scores
of all real edges not covered before time $i$.

\subsubsection{Privacy Analysis}
\label{sec:wvc-privacy}

\begin{theorem}
  \label{thm:wvc-privacy}
  The weighted vertex cover algorithm preserves $O(\epsilon)$
  differential privacy.
\end{theorem}

\newcommand{\estar}{\mathbf{e}}

\begin{proof}
  Consider some potential output $\pi$ of the private vertex cover
  algorithm, and two weighted vertex cover instances $A$ and $B$ that
  are identical except for one edge $\estar = (p, q)$. Let $p$ appear before
  $q$ in the permutation $\pi$; since the vertex sets are the same, if
  the outputs of both $A$ and $B$ are $\pi$, then $p$ will be output at
  the same time~$t$ in both executions. Let $v_t$ be the vertex output
  in Step~3 at time $t$ in such an execution; note that either $p = v_t$,
  or $p$ is output in Step~6 after $v_t$ is output.

  The probability that (conditioned on the history) a surviving
  vertex~$v$ is output in Step~3 of the algorithm at time~$i$ is:
\ifnum\fullversion=1
  \begin{gather*}
    \ts \sum_{\text{edges } e} \Pr[ \text{pick } e ] \cdot \Pr[ \text{output } v
    \mid \text{pick } e ] = \sum_{e \ni v} \frac{s(e)}{\widetilde{m}_i
      + \widetilde{n}_i/\epsilon} \cdot \frac{s(v)}{s(e)} = \frac{(d(v) +
      1/\epsilon) \cdot s(v)}{\widetilde{m}_i +
      \widetilde{n}_i/\epsilon}. \label{eq:prob-v-picked}
  \end{gather*}
\else
  \begin{gather*}
    \ts \sum_{\text{edges } e} \Pr[ \text{pick } e ] \cdot \Pr[ \text{output } v
    \mid \text{pick } e ] \\= \sum_{e \ni v} \frac{s(e)}{\widetilde{m}_i
      + \widetilde{n}_i/\epsilon} \cdot \frac{s(v)}{s(e)} = \frac{(d(v) +
      1/\epsilon) \cdot s(v)}{\widetilde{m}_i +
      \widetilde{n}_i/\epsilon}. \label{eq:prob-v-picked}
  \end{gather*}
\fi
  Since we compare the runs of the algorithm on $A$ and $B$ which differ
  only in edge $\estar$, these will be identical after time $t$ when
  $\estar$ is covered, and hence
  \begin{gather*}
    \ts \frac{\Pr[M(A) = \pi]}{\Pr[M(B) = \pi]} = \frac{(d_A(v_t) +
  1/\epsilon)s(v_t)}{(d_B(v_t) + 1/\epsilon) s(v_t)} \prod_{i \leq t}
  \left( \frac{\widetilde{m}_i^B + \widetilde{n}_i/\epsilon}{\widetilde{m}_i^A +
  \widetilde{n}_i/\epsilon} \right).
  \end{gather*}

  Note that if the extra edge $\estar \in A \setminus B$ then $d_A(v_t)
  \leq d_B(v_t) + 1$ and $\widetilde{m}_i^B \leq \widetilde{m}_i^A$, so
  the ratio of the probabilities is at most $1 + \epsilon <
  \exp(\epsilon)$.  Otherwise, the leading term is less than $1$ and
  $\widetilde{m}_i^B = \widetilde{m}_i^A + s(\estar)$, and we get
  \begin{gather*}
    \ts \frac{\Pr[M(A) = \pi]}{\Pr[M(B) = \pi]} \leq \prod_{i \leq t}
    \left(1 + \frac{s(\estar)}{\widetilde{n}_i/\epsilon}
    \right) \leq \exp \left( s(\estar) \cdot \epsilon \cdot \sum_{i \leq t}
      \frac{1}{\widetilde{n}_i} \right).
  \end{gather*}

  Let $T_j$ be the time steps $i \leq t$ where vertices in $V_j$ are
  output in $\pi$. 
  Letting $2^{j^*}$ be the weight of the lighter endpoint of edge $\estar$, we can break the sum $\sum_{i \leq t}
      \frac{1}{\widetilde{n}_i}$ into two pieces and analyze each separately:
  \begin{gather*}
    \ts \sum_{i \leq t}
      \frac{1}{\widetilde{n}_i} =
    \sum_{j \leq j^*} \sum_{i \in T_j}
      \frac{1}{\widetilde{n}_i} + \sum_{ j> j^*} \sum_{i \in T_j}
      \frac{1}{\widetilde{n}_i},
  \end{gather*}

  For the first partial sum, for some $j \leq j^*$, let $\sum_{i \in T_j} \frac{1}{\widetilde{n}_i}
  = \frac{1}{\widetilde{n}_{i_0}} + \frac{1}{\widetilde{n}_{i_1}} +
  \ldots + \frac{1}{\widetilde{n}_{i_\lambda}}$ such that $i_0 > i_1 >
  \ldots > i_\lambda$. We claim that $\widetilde{n}_{i_0} \geq 2^{-j^*}
  N_{j^*}/2$. Indeed, since $\estar$ has not yet been covered, we must
  have output fewer than $N_{j^*}/2$ vertices from levels $j^*$ or
  higher, and hence at least $N_{j^*}/2$ remaining vertices from
  $V_{j^*}$ contribute to $\widetilde{n}_{i_0}$.

  In each time step in $T_{j}$, at least one vertex of score $2^{-j}$ is
  output, so we have that $\widetilde{n}_{i_\ell} \geq 2^{-j^*} N_{j^*}/2 +
  \ell \cdot 2^{-j}$. Hence
\ifnum\fullversion=1
  \begin{gather*}
    \ts \sum_{i \in T_j} \frac{1}{\widetilde{n}_i} \leq \frac{1}{2^{-j^*}
      N_{j^*}/2} + \frac{1}{2^{-j^*}N_{j^*}/2 + 2^{-j}} + \ldots +
    \frac{1}{2^{-j^*} N_{j^*}/2 + N_j\,2^{-j}}~.
  \end{gather*}
\else
  \begin{eqnarray*}
    \ts \sum_{i \in T_j} \frac{1}{\widetilde{n}_i} &\leq& \frac{1}{2^{-j^*}
      N_{j^*}/2} + \frac{1}{2^{-j^*}N_{j^*}/2 + 2^{-j}} + \ldots \\&&+
    \frac{1}{2^{-j^*} N_{j^*}/2 + N_j\,2^{-j}}~.
  \end{eqnarray*}
\fi
  Defining $\theta = 2^{-j^* + j} \cdot N_{j^*}/2$, the expression
  above simplifies to
\ifnum\fullversion=1
  \begin{gather*}
    \ts 2^j \left( \frac{1}{\theta} + \frac{1}{\theta + 1} + \ldots +
      \frac{1}{\theta + N_j} \right) \leq 2^j \ln \left(\frac{\theta +
      N_j}{\theta}\right) = 2^j \ln \left( 1 + \frac{N_j}{\theta} \right).
  \end{gather*}
\else
  \begin{eqnarray*}
    \ts 2^j \left( \frac{1}{\theta} + \frac{1}{\theta + 1} + \ldots +
      \frac{1}{\theta + N_j} \right) &\leq& 2^j \ln \left(\frac{\theta +
      N_j}{\theta}\right) \\&=& 2^j \ln \left( 1 + \frac{N_j}{\theta} \right).
  \end{eqnarray*}

\fi
  Now using the assumption on the size of the weight classes, we have
  $N_j \leq N_{j^*} \implies N_j/\theta \leq 2^{j^* - j + 1}$, and
  hence $\sum_{i \in T_j} \frac{1}{\widetilde{n}_i} \leq (j^* - j + 2)
  2^j$, for any $j \leq j^*$.  Finally, \[ \ts \sum_{j \leq j^*} \sum_{i \in
    T_j} \frac{1}{\widetilde{n}_i} \leq \sum_{j \leq j^*} (j^* - j +
  2)2^j = O(2^{j^*}).\]

  We now consider the other partial sum $\sum_{j > j^*} \sum_{i \in T_j}
  \frac{1}{\widetilde{n}_i}$. For any such value of $i$, we know that
  $\widetilde{n_i} \geq 2^{-j^*} N_{j^*}/2$. Moreover, there are at
  most $N_{j^*}/2$ times when we output a vertex from some weight
  class $j \geq j^*$ before we output all of $V_{j^*}$; hence there are
  at most $N_{j^*}/2$ terms in the sum, each of which is at most
  $\frac{1}{2^{-j^*}\,N_{j^*}/2}$, giving a bound of $2^{j^*}$ on the
  second partial sum. Putting the two together, we get that
  \begin{gather*}
    \frac{\Pr[M(A) = \pi]}{\Pr[M(B) = \pi]} \leq \exp(s(\estar) \cdot
    \epsilon \cdot O(2^{j^*})) = \exp(O(\epsilon)),
  \end{gather*}
  using the fact that $s(\estar) \leq 2\cdot 2^{-j^*}$, since the
  lighter endpoint of $\estar$ had weight $2^{j^*}$.
\end{proof}

\subsubsection{Utility Analysis}
\label{sec:wvc-utility}

Call a vertex $v$ {\em interesting} if it is incident on a real
uncovered edge when it is picked.  Consider the weight class $V_j$: let
$I_j^1 \sse V_j$ be the set of interesting vertices output due to
Steps~3, and $I_j^2 \sse V_j$ be the set of interesting vertices of
class~$j$ output due to Step~6. The cost incurred by the algorithm is
$\sum_j 2^j (|I_j^1| + |I_j^2|)$.

\begin{lemma}
  \label{lem:wvc-stage1}
  $\E[\sum_j 2^j |I^1_j|] \leq \frac{4(1+\eps)}{\eps} \OPT$
\end{lemma}

\begin{proof}
  Every interesting vertex that our algorithm picks in Steps~3 has at
  least one real edge incident on it, and at most $\frac{1}{\eps}$
  hallucinated edges. Conditioned on selecting an interesting vertex
  $v$, the selection is due to a real edge with probability at least
  $1/(1+\frac{1}{\eps})$. One can show that the (non-private)
  algorithm $\mathcal{A}$ that selects only real edges is a
  $2$-approximation~\cite{Pitt85}. On the other hand each vertex in $I^1_j$ can be
  coupled to a step of $\mathcal{A}$ with probability $\eps/(1+\eps)$.
  Since we rounded up the costs by at most a factor of two, the claim follows.
\end{proof}

\begin{lemma}
  $\E[ |I^2_j| ] \leq 6\, \E[\sum_{j'\geq j} |I^1_{j'}|]$
\end{lemma}

\begin{proof}
  Let $t_{j}$ denote the time that class $j$ is dumped. Recall that
  by~\eqref{eq:prob-v-picked}, we pick a surviving vertex $v$ with
  probability $\propto (d(v)+\frac{1}{\eps})\cdot s(v)$ at each step.
  This expression summed over all uninteresting vertices is
  $\cup_{j'\geq j} V_{j'}$ is at most $(1/\eps)\sum_{j'\geq j}
  2^{-j'}N_{j'} \leq 2^{-j+1}N_j/\eps$. On the other hand, at each step
  before time $t_j$, all the interesting vertices in $I^2_j$ are
  available and the same expression summed over them is at least
  $2^{-j}|I^2_j|/\epsilon$. Thus for any $t \leq t_j$, conditioned on
  outputting a vertex $v_t \in \cup_{j'\geq j} V_{j'}$ in Step~3, the
  probability that it is interesting is at least
  $\frac{|I^2_j|2^{-j}/\eps}{(|I^2_j|2^{-j}+ 2^{1-j}N_j)/\eps} \geq
  \frac{|I^2_j|}{3N_j}$ (using $|I_j^2| \leq N_j$).  Now since we output
  $N_j/2$ vertices from $\cup_{j'\geq j} V_{j'}$ in Step~3 before time
  $t_j$, we conclude that $\E\big[\sum_{j'\geq j} |I^1_{j'}| \;\big|\;
  |I^2_j| \,\big] \geq \frac{N_j}{2} \times \frac{|I^2_j|}{3N_j} =
  \frac{|I^2_j|}{6}$.  Taking expectations completes the proof.
\end{proof}

We can now compute the total cost of all the interesting vertices dumped
in Steps~6 of the algorithm.
\ifnum\fullversion=1
\begin{align*}
  \label{eq:second-to-first}
  \ts \E[\textsf{cost}(\bigcup_j I_j^2)] &= \ts \sum_j 2^j\, \E[|I_j^2|] \leq  6
  \;\sum_j 2^j \; \sum_{j' \geq j} \E[ |I_{j'}^1| ]  \ts \leq 6 \;\sum_{j'} \E[|I_{j'}^1|] \; 2^{j'+1}
    \leq  12\; \cdot \E[\textsf{cost}(\bigcup_j I_j^1)].
\end{align*}
\else
\begin{eqnarray*}
  \label{eq:second-to-first}
  \ts \E[\textsf{cost}(\bigcup_j I_j^2)] &=& \ts \sum_j 2^j\, \E[|I_j^2|] \\&\leq&  6
  \;\sum_j 2^j \; \sum_{j' \geq j} \E[ |I_{j'}^1| ]\\  \ts &\leq& 6 \;\sum_{j'} \E[|I_{j'}^1|] \; 2^{j'+1}
   \\& \leq&  12\; \cdot \E[\textsf{cost}(\bigcup_j I_j^1)].
\end{eqnarray*}
\fi
Finally, combining this calculation with \lref[Lemma]{lem:wvc-stage1}, we
conclude that our algorithm gives an $O(\frac{1}{\eps})$ approximation
to the weighted vertex cover problem.

\subsection{Vertex Cover Lower Bounds}
\label{sec:vc-lower-bound}

\begin{theorem}
  \label{thm:vc-lbd}
  Any algorithm for the vertex cover problem that prescribes
  edge-orientations with $\epsilon$-differential privacy must have an
  $\Omega(1/\epsilon)$ approximation guarantee, for any $\epsilon \in (\frac{1}{n},1]$.
\end{theorem}

\begin{proof}
  Let $V= \{1, 2, \ldots, \lceil\frac{1}{2\epsilon}\rceil\}$, and let
  $M$ be an $\epsilon$-differentially private algorithm that takes as
  input a private set $E$ of edges, and outputs an orientation $M_E:
  V\times V \rightarrow V$, with $M_E(u,v) \in \{u,v\}$ indicating to the edge which
  endpoint to use.
  Picking two distinct vertices $u \neq v$ uniformly at random (and
  equating $(u,v)$ with $(v,u)$), we have by symmetry:
  \[\ts \Pr_{u,v}[M_\emptyset((u,v)) \neq u] = \frac{1}{2}.\]
  Let $\star_u = (V, \{u\}\times (V\setminus \{u\}))$ be the star graph
  rooted at $u$.
  Since $\star_u$ and $\emptyset$ differ in at most
  $\frac{1}{2\epsilon} - 1 < \frac{1}{\epsilon}$ edges and $M$ satisfies $\epsilon$-differential
  privacy, we conclude that
  \[\ts \Pr_{u,v}[M_{\star_u}((u,v)) \neq u] \ge \frac{1}{2e}.\]
  Thus the expected cost of $M$ when input a uniformly random $\star_u$ is at least
  $\frac{1}{2e}\times \lceil\frac{1}{2\epsilon}\rceil$, while
  $\OPT(\star_u)$ is 1. We can repeat this
  pattern arbitrarily, picking a random star from each group of
  $1/\epsilon$ vertices; this results in graphs with arbitrarily large
  vertex covers where $M$ incurs cost $1/\epsilon$ times the cost.
\end{proof}

\section{Set Cover}
\label{sec:SC}
We now turn our attention to private approximations for the Set Cover Problem;
here the set system $(U, \mathcal{S})$ is public, but the actual set of
elements to be covered $R \sse U$ is the private information.  As for
vertex cover, we cannot explicitly output a set cover that is good and
private at the same time.  Hence, we again output a permutation over all
the sets in the set system; this implicitly defines a set cover for $R$
by picking, for each element $R$, the first set in this permutation that
contains it. Our algorithms for set cover give the slightly weaker
$(\epsilon, \delta)$-privacy guarantees. 

\subsection{Unweighted Set Cover}
\label{sec:unweighted-set-cover}

We are given a set system $(U,\mathcal{S})$
and must cover a private subset $R
\subset U$.
Let the cardinality of the set
system be $|\mathcal{S}| = m$, and let $|U| = n$.
We first observe a computationally inefficient algorithm.

\begin{theorem}
  The exponential mechanism, when used to pick a permutation of sets,
  runs in time $O(m!poly(n))$ and gives an $O(\log
  (em/\OPT)/\epsilon)$-approximation.
\end{theorem}
\begin{proof}
  A random permutation, with probability at least ${m\choose \OPT}^{-1}$
  has all the sets in $\OPT$ before any set in $\OPT^{c}$. Thus the
  additive error is $O(\log {m\choose \OPT}/\epsilon)$.
\end{proof}
The rest of the section gives a computationally efficient algorithm with
slightly worse guarantees: this is a modified version of the greedy
algorithm, using the exponential mechanism to bias towards picking large
sets.
\begin{algorithm}
  \caption{Unweighted Set Cover}
  \begin{algorithmic}[1]
    \STATE \textbf{Input:} Set system $(U,\mathcal{S})$, private $R
    \subset U$ of elements to cover, $\epsilon$,$\delta$.
    \STATE \textbf{let} $i \leftarrow 1$, $R_i = R$, $\mathcal{S}_i
    \leftarrow \mathcal{S}$. $\epsilon' \leftarrow \epsilon/2\ln(\frac{e}{\delta})$.
    \FOR{$i=1,2,\ldots,m$}
    \STATE \textbf{pick} a set $S$ from $\mathcal{S}_i$ with probability
    proportional to $\exp(\epsilon' |S \cap R_i|)$.
    \STATE \textbf{output} set $S$.
    \STATE $R_{i+1} \leftarrow R_i \setminus S$, $\mathcal{S}_{i+1} \leftarrow \mathcal{S}_i - \{S\}$.
    \ENDFOR
\end{algorithmic}
\end{algorithm}

\subsubsection{Utility Analysis}
\label{sec:unwtd-sc-utility}

At the beginning of iteration~$i$, say there are $m_i = m - i + 1$
remaining sets and $n_i = |R_i|$ remaining elements, and define $L_i =
\max_{S \in \mathcal{S}} |S \cap R_i|$, the largest number of uncovered
elements covered by any set in $\mathcal{S}$.
By a standard argument, any algorithm that always picks
sets of size $L_i/2$ is an $O(\ln n)$ approximation algorithm.

\begin{theorem}
  The above algorithm achieves an expected approximation ratio of
  $O(\ln n + \frac{\ln m}{\epsilon'}) = O(\ln n + \frac{\ln m \ln (e/\delta)}{\epsilon}) $.
\end{theorem}

\begin{proof}
  As there is at least one set containing $L_i$ elements, our use of the exponential mechanism to select sets combined with Equation \ref{eqn:expmech} ensures that the probability we select a set covering fewer than $L_i - 3\ln m/\epsilon$  elements is at most $1/m^2$. While $L_i > 6\ln m/\epsilon$, with probability at least $(1-1/m)$ we always select sets that cover at least $L_i/2$ elements, and can therefore use no more than $O(\OPT \ln n)$ sets. Once $L_i$ drops below this bound, we observe that the number of remaining elements $|R_i|$ is at most $\OPT \cdot L_i$. Any permutation therefore costs at most an additional $O(\OPT \ln m / \epsilon')$.
\end{proof}

\subsubsection{Privacy}
\label{sec:unwtd-sc-privacy}

\begin{theorem}
\label{thm:unwtdsetcoverprivacy}
  The unweighted set cover algorithm preserves
  $(\epsilon,\delta)$ differential privacy for any $\epsilon \in (0,1)$, and $\delta < 1/e$.
\end{theorem}

\begin{proof}
  Let $A$ and $B$ be two set cover instances that differ in some element
  $I$.  Say that $S^I$ is the
  collection of sets containing $I$. Fix an output permutation $\pi$,
  and write $s_{i,j}(A)$ to denote the size of set $S_j$ after the first
  $i-1$ sets in $\pi$ have been added to the cover.
\ifnum\fullversion=1
  \begin{eqnarray*}
    \frac{\Pr[M(A)=\pi]}{\Pr[M(B)=\pi]} &=& \prod_{i=1}^n\left(\frac{\exp(\epsilon'\cdot s_{i,\pi_i}(A))/(\sum_j\exp(\epsilon'\cdot s_{i,j}(A)))}{\exp(\epsilon'\cdot s_{i,\pi_i}(B))/(\sum_j\exp(\epsilon'\cdot s_{i,j}(B)))}\right) \\
    &=& \frac{\exp(\epsilon'\cdot s_{t,\pi_t}(A))}{\exp(\epsilon'\cdot s_{t,\pi_t}(B))}\cdot \prod_{i=1}^t\left(\frac{\sum_j\exp(\epsilon'\cdot s_{i,j}(B))}{\sum_j\exp(\epsilon'\cdot s_{i,j}(A))}\right)
  \end{eqnarray*}
\else
  \begin{align*}
    &\frac{\Pr[M(A)=\pi]}{\Pr[M(B)=\pi]} \\
    &= \prod_{i=1}^n\left(\frac{\exp(\epsilon'\cdot s_{i,\pi_i}(A))/(\sum_j\exp(\epsilon'\cdot s_{i,j}(A)))}{\exp(\epsilon'\cdot s_{i,\pi_i}(B))/(\sum_j\exp(\epsilon'\cdot s_{i,j}(B)))}\right) \\
    &= \frac{\exp(\epsilon'\cdot s_{t,\pi_t}(A))}{\exp(\epsilon'\cdot s_{t,\pi_t}(B))}\cdot \prod_{i=1}^t\left(\frac{\sum_j\exp(\epsilon'\cdot s_{i,j}(B))}{\sum_j\exp(\epsilon'\cdot s_{i,j}(A))}\right)
  \end{align*}
\fi
  where $t$ is such that $S_{\pi_t}$ is the first set containing $I$ to
  fall in the permutation $\pi$. After $t$, the remaining
  elements in $A$ and $B$ are identical, and all subsequent terms
  cancel. Moreover, except for the $t^{th}$ term, the numerators of both
  the top and bottom expression cancel, since all the relevant set sizes
  are equal. If $A$ contains $I$ and $B$ does not the first term is $\exp(\epsilon')$ and the each term in the product is at most 1.

  Now suppose that $B$ contains $I$ and $A$ does not . In this case, the first term is $\exp(-\epsilon') < 1$. Moreover, in instance $B$, every set in $S^I$ is
  larger by 1 than in $A$, and all others remain the same size.
  Therefore, we have:
\ifnum\fullversion=1
  \begin{eqnarray*}
    \frac{\Pr[M(A)=\pi]}{\Pr[M(B)=\pi]} &\leq&  \prod_{i=1}^t\left(\frac{(\exp(\epsilon')-1)\cdot\sum_{j\in S^I}\exp(\epsilon'\cdot s_{i,j}(A)) + \sum_j\exp(\epsilon'\cdot s_{i,j}(A))}{\sum_j\exp(\epsilon'\cdot s_{i,j}(A))}\right) \\
    &=&  \prod_{i=1}^t \left(1 + (\exp(\epsilon')-1)\cdot p_i(A) \right)
  \end{eqnarray*}
\else
  \begin{align*}
    &\ts\frac{\Pr[M(A)=\pi]}{\Pr[M(B)=\pi]} \\
    &\ts \leq  \prod_{i=1}^t\left(\frac{(\exp(\epsilon')-1)\cdot\sum_{j\in S^I}\exp(\epsilon'\cdot s_{i,j}(A)) + \sum_j\exp(\epsilon'\cdot s_{i,j}(A))}{\sum_j\exp(\epsilon'\cdot s_{i,j}(A))}\right) \\
    & \ts =  \prod_{i=1}^t \left(1 + (\exp(\epsilon')-1)\cdot p_i(A) \right)
  \end{align*}\fi
  where $p_i(A)$ is the probability that a set containing $I$ is chosen
  at step $i$ of the algorithm running on instance $A$, conditioned on
  picking the sets $S_{\pi_1}, \ldots, S_{\pi_{i-1}}$ in the previous
  steps.

For an instance $A$ and an element $I\in A$, we say that an output $\sigma$ is {\em $q$-bad} if $\sum_{i} p_i(A) \one(I \mbox{ uncovered at step } i)$ (strictly) exceeds $q$, where $p_i(A)$ is as defined above. We call a permutation {\em $q$-good} otherwise. We first consider the case when the output $\pi$ is $(\ln \delta^{-1})$-good. By the definition of $t$, we have
 \begin{equation*}
     \sum_{i=1}^{t-1}p_i(A) \leq \ln \delta^{-1}.
  \end{equation*}
Continuing the analysis from above,
  \begin{align*}
    \frac{\Pr[M(A)=\pi]}{\Pr[M(B)=\pi]} &\leq
    \prod_{i=1}^t \exp((\exp(\epsilon')-1)p_i(A))
    \leq \exp(2\epsilon'\sum_{i=1}^tp_i(A)) \\
    &\leq \exp(2\epsilon'(\ln(\frac{1}{\delta})+p_t(A)))
    \leq \exp(2\epsilon'(\ln(\frac{1}{\delta})+1)).
  \end{align*}
Thus, for any $(\ln \delta^{-1})$-good output $\pi$, we have $\frac{\Pr[M(A)=\pi]}{\Pr[M(B)=\pi]} \leq \exp(\epsilon)$.
We can then invoke the following lemma, proved in appendix~\ref{sec:deltafix}
\begin{lemma}
\label{lem:deltafix}
For any set system $(U,\mathcal{S})$, any instance $A$ and any $I \in A$, the probability that the output $\pi$ of the algorithm above is $q$-bad is bounded by $\exp(-q)$.
\end{lemma}

Thus for any set $\mathcal{P}$ of outcomes, we have
\ifnum\fullversion=1
\begin{eqnarray*}
\Pr[M(A) \in \mathcal{P}] &=& \sum_{\pi \in \mathcal{P}} \Pr[M(A)=\pi]\\
&=& \sum_{\pi \in \mathcal{P}: \pi\mbox{ is } (\ln \delta^{-1})\mbox{-good}} \Pr[M(A)=\pi] + \sum_{\pi \in \mathcal{P}: \pi\mbox{ is } (\ln \delta^{-1})\mbox{-bad} } \Pr[M(A)=\pi] \\
&\leq& \sum_{\pi \in \mathcal{P}: \pi\mbox{ is } (\ln \delta^{-1})\mbox{-good} } \exp(\epsilon) \Pr[M(B)=\pi] + \delta\\
&\leq& \exp(\epsilon)\Pr[M(B) \in \mathcal{P}] +\delta.
\end{eqnarray*}
\else
\begin{align*}
&\Pr[M(A) \in \mathcal{P}]\\
&= \sum_{\pi \in \mathcal{P}} \Pr[M(A)=\pi]\\
&= \sum_{\pi \in \mathcal{P}: \pi\mbox{ is } (\ln \delta^{-1})\mbox{-good}} \Pr[M(A)=\pi] \\
&\;\;\;\;+ \sum_{\pi \in \mathcal{P}: \pi\mbox{ is } (\ln \delta^{-1})\mbox{-bad} } \Pr[M(A)=\pi] \\
&\leq \sum_{\pi \in \mathcal{P}: \pi\mbox{ is } (\ln \delta^{-1})\mbox{-good} } \exp(\epsilon) \Pr[M(B)=\pi] + \delta\\
&\leq \exp(\epsilon)\Pr[M(B) \in \mathcal{P}] +\delta.
\end{align*}
\fi
\end{proof}

\begin{corollary}
  For $\epsilon < 1$ and $\delta = 1/\mathrm{poly}(n)$, there is an $O(\frac{\ln n\ln m}{\epsilon})$-approximation algorithm for the unweighted set cover problem
  preserving $(\epsilon,\delta)$-differential privacy.
\end{corollary}

\subsection{Weighted Set Cover}

\newcommand{\poly}{\operatorname{poly}}

We are given a set system $(U,\mathcal{S})$ and a cost function
$C:\mathcal{S}\rightarrow \mathbb{R}$. We must cover a private subset
$R \subset U$. W.l.o.g., let $\min_{S \in \mathcal{S}}C(S) = 1$, and denote
$\max_{S \in \mathcal{S}}C(S) = W$.  Let the cardinality of the set
system be $|\mathcal{S}| = m$, and let $|U| = n$.

\begin{algorithm}[H]
  \caption{Weighted Set Cover}
  \begin{algorithmic}[1]
    \STATE \textbf{let} $i \leftarrow 1$, $R_i = R$, $\mathcal{S}_i
    \leftarrow \mathcal{S}$, $r_i \leftarrow n$, $\epsilon'=\frac{\epsilon}{2\ln(e/\delta)}$, $T =  \Theta\big(\frac{\log m + \log\log(nW)}{\epsilon'}\big)$

    \WHILE{$r_i \geq 1/W$}

    \STATE \textbf{pick} a set $S$ from $\mathcal{S}_i$ with probability
    proportional to $\exp\big(\epsilon' \big(\,|S \cap R_i| - r_i\cdot C(S)\,\big)\big)$ \\
    ~~~~~~~~or \texttt{halve} with probability proportional to $\exp(-\epsilon' T)$
    \IF{\texttt{halve} }
    \STATE \textbf{let} $r_{i+1} \leftarrow r_i/2$, $R_{i+1} \leftarrow R_i$,
    $\mathcal{S}_{i+1} \leftarrow \mathcal{S}_i$, $i \gets i+1$
    \ELSE
    \STATE \textbf{output} set $S$
    \STATE \textbf{let} $R_{i+1} \leftarrow R_i \setminus S$,
    $\mathcal{S}_{i+1} \leftarrow \mathcal{S}_i - \{S\}$, $r_{i+1} \leftarrow r_i$,
    $i \leftarrow i+1$
    \ENDIF
    \ENDWHILE
    \STATE \textbf{output} all remaining sets in $\mathcal{S}_i$ in
    random order
  \end{algorithmic}
\end{algorithm}

Let us first analyze the utility of the algorithm. If $R=\emptyset$, the algorithm has cost zero and there is nothing to prove. So we can assume that $\OPT \geq 1$. We first show that (\whp) $r_i \gtrapprox R_i/\OPT$.
\begin{lemma}
  \label{firstSetCoverLemma}
  Except with probability $1/\poly(m)$, we have
  $r_i \geq \frac{|R_i|}{2 \OPT}$ for all iterations~$i$.
\end{lemma}
\begin{proof}
  Clearly $r_1 = n \geq |R_1|/2\OPT$. For $r_i$ to fall below
  $|R_i|/2$, it must be in $(\frac{|R_i|}{2 \OPT},\frac{|R_i|}{ \OPT}]$
  and be halved in Step~6 of some iteration $i$. We'll show that this is
  unlikely: if at some iteration $i$, $\frac{|R_i|}{2 \OPT} \leq r_i
  \leq \frac{|R_i|}{\OPT}$, then we argue that with high probability, the algorithm
  will not output \texttt{halve} and thus not halve $r_i$. Since all
  remaining elements $R_i$ can be covered at cost at most $\OPT$, there
  must exist a set $S$ such that $\frac{|S \cap R_i|}{C(S)} \geq
  \frac{|R_i|}{\OPT}$, and hence $|S \cap R_i| \geq C(S) \cdot \frac{
    |R_i|}{\OPT} \geq C(S) \cdot r_i.$

  Hence $u_i(S) := |S \cap R_i| - r_i\cdot C(S) \geq 0$ in this case,
  and the algorithm will output $S$ with probability at least
  proportional to $1$, whereas it outputs \texttt{halve} with probability
  proportional to $\exp(-\epsilon' T)$. Thus, $\Pr[\; \text{algorithm
    returns \texttt{halve}}\; ] < \exp(-\epsilon' T) = 1/\poly(m \log
  nW)$.  Since there are $m$ sets in total, and $r$ ranges from  $n$ 
  to $1/W$, there are at most $m + O(\log n W)$ iterations, and the
  proof follows by a union bound.
\end{proof}

Let us define a \emph{score} function $u_i(S) := |S \cap R_i| - r_i\cdot
C(S)$, and $u_i(\texttt{halve}) := -T$: note that in Step~4 of our
algorithm, we output either \texttt{halve} or a set $S$, with
probabilities proportional to $\exp(\epsilon' u_i(\cdot))$.  The following
lemma states that with high probability, none of the sets output by our
algorithm have very low scores (since we are much more likely to output
\texttt{halve} than a low-scoring set).
\begin{lemma}
  \label{goodAnswerLemma}
  Except with probability at most $1/\poly(m)$,
  Step~4 only returns sets $S$ with $u_i(S) \geq -2T$.
\end{lemma}
\begin{proof}
  There are at most $|\mathcal{S}_i| \leq m$ sets $S$ with score $u_i(S)
  \leq -2T$, and so one is output with probability at most proportional
  to $m\exp(-2T\epsilon)$. We will denote this bad event by
  $\mathcal{B}$. On the other hand, \texttt{halve} is output with
  probability proportional to $\exp(-T\epsilon)$. Hence,
  $\Pr[\texttt{halve}]/\Pr[\mathcal{B}] \geq \exp(T\epsilon)/m$, and so
  $\Pr[\mathcal{B}] \leq m/\exp(T\epsilon) \leq 1/\poly(m \log nW)$.
  Again there are at most $m + O(\log nW)$ iterations, and the lemma
  follows by a trivial union bound.
\end{proof}

We now analyze the cost incurred by the algorithm in each stage. Let us
divide the algorithm's execution into \emph{stages}: stage~$j$ consists
of all iterations $i$ where $|R_i| \in (\frac{n}{2^j},
\frac{n}{2^{j-1}}]$.  Call a set $S$ {\textit interesting} if it is
incident on an uncovered element when it is picked.  Let $\mathcal{I}_j$
be the set of interesting sets selected in stage $j$, and
$C(\mathcal{I}_j)$ be the total cost incurred on these sets.

\begin{lemma}
  \label{secondSetCoverLemma}
  Consider stages $1,\ldots,j$ of the algorithm.  Except with probability $1/\poly(m)$, we can bound
  the cost of the interesting sets in stage~$1,\ldots,j$ by:
  $$\sum_{j' \leq j} C(\mathcal{I}_{j'}) \leq 4j \OPT\cdot(1+2T).$$
\end{lemma}

\begin{proof}
  By \lref[Lemma]{goodAnswerLemma} all the output sets have $u_i(S_i)
  \geq -2T$ \whp. Rewriting, each $S_i$ selected in a round $j' \leq
  j$ satisfies
  \begin{gather*}
     C(S_i) \leq \frac{|S_i \cap R_i| + 2T}{r_i} \leq
    \frac{2^{j'+1}\;\OPT}{n}(|S_i \cap R_i| + 2T),
  \end{gather*}
  where the second inequality is \whp, and uses
  \lref[Lemma]{firstSetCoverLemma}.  Now summing over all rounds $j' \leq
  j$, we get
\ifnum\fullversion=1
  \begin{gather*}
    \label{eq:wsc-1}
    \sum_{j' \leq j} C(\mathcal{I}_{j'}) \leq \sum_{j' \leq j}
    \frac{2^{j'+1}\;\OPT}{n} \big( \sum_{i \textrm{ s.t. }
      S_i \in \mathcal{I}_{j'}} \big( |S_i \cap R_i| + 2T \big) \big).
  \end{gather*}
\else
  \begin{align}
    \label{eq:wsc-1}
    &\sum_{j' \leq j} C(\mathcal{I}_{j'})\\
    &\leq \sum_{j' \leq j}
    \frac{2^{j'+1}\;\OPT}{n} \big( \sum_{i \textrm{ s.t. }
      S_i \in \mathcal{I}_{j'}} \big( |S_i \cap R_i| + 2T \big) \big).
  \end{align}
\fi
  Consider the inner sum for any particular value of $j'$: let the first
  iteration in stage $j'$ be iteration~$i_0$---naturally $R_i \sse
  R_{i_0}$ for any iteration $i$ in this stage. Now, since $S_i \cap R_i
  \sse R_{i_0}$ and $S_i \cap R_i$ is disjoint from $S_{i'} \cap
  R_{i'}$, the sum over $|S_i \cap R_i|$ is at most $|R_{i_0}|$, which
  is at most $\frac{n}{2^{j'-1}}$ by definition of
  stage $j'$.  Moreover, since we are only concerned with bounding the
  cost of interesting sets, each $|S_i \cap R_i| \geq 1$, and so $|S_i
  \cap R_i| + 2T \leq |S_i \cap R_i|(1+2T)$. Putting this
  together,~(\ref{eq:wsc-1}) implies
\ifnum\fullversion=1
  \begin{gather*}
    \sum_{j' \leq j} C(\mathcal{I}_{j'}) \leq \sum_{j' \leq j}
    \frac{2^{j'+1}\;\OPT}{n} \times
    \frac{n}{2^{j'-1}} (1+2T) = 4j\,\OPT\,(1+2T),
  \end{gather*}
\else
  \begin{eqnarray*}
    \sum_{j' \leq j} C(\mathcal{I}_{j'})
    &\leq& \sum_{j' \leq j}
    \frac{2^{j'+1}\;\OPT}{n} \times
    \frac{n}{2^{j'-1}} (1+2T)\\
    &=& 4j\,\OPT\,(1+2T),
  \end{eqnarray*}
\fi
  which proves the lemma.
\end{proof}

\begin{theorem}[\textbf{Utility}]
  The weighted set cover algorithm incurs a cost of
 $O(T\,\log n\,\OPT)$
  except with probability $1/\poly(m)$.
\end{theorem}
\begin{proof}
  Since the number of uncovered elements halves in each stage by
  definition, there are at most $1 + \log n$ stages, which by
  \lref[Lemma]{secondSetCoverLemma} incur a total cost of at most $O(\log
  n\;\OPT\cdot(1+2T))$. The sets that remain and are output at
  the very end of the algorithm incur cost at most $W$ for each
  remaining uncovered element; since $r_i < 1/W$ at the end,
  \lref[Lemma]{firstSetCoverLemma} implies that $|R_i| < 2 \OPT /W$ (\whp),
  giving an additional cost of at most $2\, \OPT$.
\end{proof}

We can adapt the above argument to bound the expected cost by $O(T\log n
\;\OPT)$. 

\begin{theorem}[\textbf{Privacy}]
\label{lem:wtdsetcoverprivacy}
For any $\delta > 0$, the weighted set cover algorithm preserves
$\left(\epsilon,\delta\right)$
differential privacy.
\end{theorem}

\begin{proof}
We imagine that the algorithm outputs a set
  named ``HALVE'' when Step~4 of the algorithm returns \texttt{halve}, and
  show that even this output is privacy preserving. 
  Let $A$ and $B$ be two set cover instances that differ in some element
  $I$.  Say that $S^I$ is the
  collection of sets containing $I$. Fix an output $\pi$,
  and write $u_{i,j}(A)$ to denote the score of $\pi_j$ (recall this
  may be \texttt{halve}) after the first
  $i-1$ sets in $\pi$ have been selected.
\ifnum\fullversion=1
  \begin{displaymath}
    \frac{\Pr[M(A)=\pi]}{\Pr[M(B)=\pi]} =\prod_{i=1}^n\left(\frac{\exp(\epsilon'\cdot u_{i,\pi_i}(A))/(\sum_j\exp(\epsilon'\cdot u_{i,j}(A)))}{\exp(\epsilon'\cdot u_{i,\pi_i}(B))/(\sum_j\exp(\epsilon'\cdot u_{i,j}(B)))}\right)
    = \frac{\exp(\epsilon'\cdot u_{t,\pi_t}(A))}{\exp(\epsilon'\cdot u_{t,\pi_t}(B))}\cdot \prod_{i=1}^t\left(\frac{\sum_j\exp(\epsilon'\cdot u_{i,j}(B))}{\sum_j\exp(\epsilon'\cdot u_{i,j}(A))}\right)
  \end{displaymath}
\else
  \begin{align*}
    &\frac{\Pr[M(A)=\pi]}{\Pr[M(B)=\pi]} \\
    &=\prod_{i=1}^n\left(\frac{\exp(\epsilon'\cdot u_{i,\pi_i}(A))/(\sum_j\exp(\epsilon'\cdot u_{i,j}(A)))}{\exp(\epsilon'\cdot u_{i,\pi_i}(B))/(\sum_j\exp(\epsilon'\cdot u_{i,j}(B)))}\right)\\
    &= \frac{\exp(\epsilon'\cdot u_{t,\pi_t}(A))}{\exp(\epsilon'\cdot u_{t,\pi_t}(B))}\cdot \prod_{i=1}^t\left(\frac{\sum_j\exp(\epsilon'\cdot u_{i,j}(B))}{\sum_j\exp(\epsilon'\cdot u_{i,j}(A))}\right)
  \end{align*}
\fi
  where $t$ is such that $S_{\pi_t}$ is the first set containing $I$ to
  fall in the permutation $\pi$. After $t$, the remaining
  elements in $A$ and $B$ are identical, and all subsequent terms
  cancel. Moreover, except for the $t^{th}$ term, the numerators of both
  the top and bottom expression cancel, since all the relevant set sizes
  are equal. If $A$ contains $I$ and $B$ does not the first term is $\exp(\epsilon')$ and the each term in the product is at most 1. Since $\epsilon' \leq \epsilon$, we conclude that in this case, for any set $\mathcal{P}$ of outputs, $\Pr[M(A)\in \mathcal{P}] \leq \exp(\epsilon)\Pr[M(B)\in \mathcal{P}]$.

  Now suppose that $B$ contains $I$ and $A$ does not . In this case, the first term is $\exp(-\epsilon') < 1$. Moreover, in instance $B$, every set in $S^I$ is
  larger by 1 than in $A$, and all others remain the same size.
  Therefore, we have:
  \begin{displaymath}
    \frac{\Pr[M(A)=\pi]}{\Pr[M(B)=\pi]} \leq  \prod_{i=1}^t\left(\frac{(\exp(\epsilon')-1)\cdot\sum_{j\in S^I}\exp(\epsilon'\cdot u_{i,j}(A)) + \sum_j\exp(\epsilon'\cdot u_{i,j}(A))}{\sum_j\exp(\epsilon'\cdot u_{i,j}(A))}\right)
    =  \prod_{i=1}^t \left(1 + (e^{\epsilon'}-1)\cdot p_i(A) \right)
  \end{displaymath}
  where $p_i(A)$ is the probability that a set containing $I$ is chosen
  at step $i$ of the algorithm running on instance $A$, conditioned on
  picking the sets $S_{\pi_1}, \ldots, S_{\pi_{i-1}}$ in the previous
  steps.

For an instance $A$ and an element $I\in A$, we say that an output $\sigma$ is {\em $q$-bad} if $\sum_{i} p_i(A) \one(I \mbox{ uncovered at step } i)$ (strictly) exceeds $q$, where $p_i(A)$ is as defined above. We call a permutation {\em $q$-good} otherwise. We first consider the case when the output $\pi$ is $(\ln \delta^{-1})$-good. By the definition of $t$, we have
 \begin{equation*}
     \sum_{i=1}^{t-1}p_i(A) \leq \ln \delta^{-1}.
  \end{equation*}
Continuing the analysis from above,
  \begin{align*}
    \frac{\Pr[M(A)=\pi]}{\Pr[M(B)=\pi]} &\leq
    \prod_{i=1}^t \exp((\exp(\epsilon')-1)p_i(A))
    \leq \exp\left(2\epsilon'\sum_{i=1}^tp_i(A)\right) \\
    & \leq \exp\left(2\epsilon'\left(\ln \delta^{-1} +p_t(A)\right)\right)
    \leq \exp\left(2\epsilon'\left(\ln \delta^{-1} +1\right)\right).
  \end{align*}
Thus, for any $(\ln \delta^{-1})$-good output $\pi$, we have $\frac{\Pr[M(A)=\pi]}{\Pr[M(B)=\pi]} \leq \exp(\epsilon)$.

Finally, as in the proof of Theorem~\ref{thm:unwtdsetcoverprivacy}, we can use lemma~\ref{lem:deltafix} to complete the proof.
\end{proof}

\subsection{Removing the Dependence on $W$}

\newcommand{\I}{\mathcal{I}}

We can remove the dependence of the algorithm on $W$ with a simple idea.
For an instance $\I = (U, \ess)$, let $\ess^j = \{S \in \ess \mid C(S)
\in (n^{j},n^{j+1}]\,\}$. Let $U^j$ be the set of elements such that the
cheapest set containing them is in $\ess^j$. Suppose that for each $j$
and each $S \in \ess^j$, we remove all elements that can be covered by a
set of cost at most $n^{j-1}$, and hence define $S'$ to be $S \cap (U^j
\cup U^{j-1})$. This would change the cost of the optimal solution only
by a factor of $2$, since if we were earlier using $S$ in the optimal
solution, we can pick $S'$ and at most $n$ sets of cost at most
$n^{j-1}$ to cover the elements covered by $S \setminus S'$. Call this
instance $\I' = (U, \ess')$.

Now we partition this instance into two instances $\I_1$ and $\I_2$,
where $\I_1 = (\cup_{j \text{ even}} U^j, \ess')$, and where $\I_2 =
(\cup_{j \text{ odd}} U^j, \ess')$. Since we have just partitioned the
universe, the optimal solution on both these instances costs at most
$2\,\OPT(\I)$. But both these instances $\I_1, \I_2$ are themselves
collections of \emph{disjoint} instances, with each of these instances
having $w_{\max}/w_{\min} \leq n^2$; this immediately allows us to
remove the dependence on $W$. Note that this transformation is
based only on the set system $(U, \ess)$, and not on the private
subset~$R$.

\begin{theorem}
For any $\epsilon\in (0,1)$, $\delta = 1/\poly(n)$, there is an $O(log n (\log m + \log\log n)/\epsilon)$-approximation for the weighted set cover problem that preserves $(\epsilon,\delta)$-differential privacy.
\end{theorem}

\subsection{Lower bounds}
\begin{theorem}
Any $\epsilon$-differentially private algorithm that maps elements to
sets must have approximation factor $\Omega(\log m/\epsilon)$, for a
set cover instance with $m$ sets and $((\log m)/\epsilon)^{O(1)}$
elements, for any $\epsilon \in (2\log m/m^{\frac{1}{20}}, 1)$.
\end{theorem}

\begin{proof}
  We consider a set system with $|U| = N$ and $\mathcal{S}$ a uniformly
  random selection of $m$ size-$k$ subsets of $U$.  We will consider
  problem instances $S_i$ consisting of one of these $m$ subsets, so
  $\OPT(S_i) = 1$. Let $M$ be an $\epsilon$-differentially private
  algorithm that on input $T\subseteq U$, outputs an assignment $f$
  mapping each element in $U$ to some set in $\mathcal{S}$ that covers
  it.  The number of possible assignments is at most $m^N$.  The cost on
  input $T$ under an assignment $f$ is the cardinality of the set $f(T)
  = \cup_{e\in T} f(e)$.

  We say assignment $f$ is good for a subset $T\subseteq U$ if its cost
  $\card{f(T)}$ is at most $l = \frac{k}{2}$. We first show that any
  fixed assignment $f: U \rightarrow [m]$, such that $\card{f^{-1}(j)}
  \leq k$ for all $j$, is unlikely to be good for a randomly picked
  size-$k$ subset $T$ of $U$.  The number of ways to choose $l$ sets
  from among those with non-empty $f^{-1}(\cdot)$ is at most
  $\binom{N}{l}$.  Thus the probability that $f$ is good for a random
  size-$k$ subset is at most $\binom{N}{l} \left(\frac{lk}{N}\right)^k$.
  Setting $k = N^{1/10}$, and $l=\frac{k}{2}$, this is at most
  \[\left(\frac{N e}{l}\right)^l \left(\frac{lk}{N}\right)^k =
  \left(\frac{e k^3}{2N}\right)^{k/2} \leq 2^{-k \log N / 4}.\] Let
  $m=2^{2\epsilon k}$. The probability that $f$ is good for at least $t$
  of our $m$ randomly picked sets is bounded by
  \[\binom{m}{t} \left(2^{-k \log N /4}\right)^t \leq 2^{2\epsilon kt} 2^{-tk \log
    N /4} \leq 2^{-tk \log k /8}.\] Thus, with probability at most
  $2^{-N k\log k/8}$, a fixed assignment is good for more than $N$ of
  $m$ randomly chosen size-$k$ sets.  Taking a union bound over $m^N =
  2^{2\epsilon kN}$ possible assignments, the probability that {\em any}
  feasible assignment $f$ is good for more than $N$ sets is at most
  $2^{-Nk\log k/16}$. Thus there exists a selection of size-$k$ sets
  $S_1,\ldots,S_m$ such that no feasible assignment $f$ is good for more
  than $N$ of the $S_i$'s.

  Let $p_{M(\emptyset)}(S_i)$ be the probability that an assignment
  drawn from the distribution defined by running $M$ on the the empty
  set as input is good for $S_i$. Since any fixed assignment is good for
  at most $N$ of the $m$ sets, the average value of $p_{M(\emptyset)}$
  is at most $N/m$. Thus there exists a set, say $S_1$ such that
  $p_{M(\emptyset)}(S_1) \leq N/m$. Since $\card{S_i} = k$ and $M$ is
  $\epsilon$-differentially private, $p_{M(S_1)}(S_1) \leq exp(\epsilon
  k) p_{M(\emptyset)}(S_1) < \frac{1}{2}$. Thus with probability at
  least half, the assignment $M$ picks on $S_1$ is not good for $S_1$.
  Since $\OPT(S_1) = 1$, the expected approximation ratio of $M$ is at
  least $l/2 = \frac{\log m}{4\epsilon}$.

  Additionally, one can take $s$ distinct instances of the above
  problem, leading to a new instance on $s\cdot N$ elements and $s\cdot
  m$ sets. $\OPT$ is now $s$, while it is easy to check that any private
  algorithm must cost $\Omega(s\cdot l)$ in expectation. Thus the lower
  bound in fact rules out additive approximations.
\end{proof} 

\subsection{An Inefficient Algorithm for Weighted Set Cover}
\newcommand{\mycost}{\widetilde{cost}}
\newcommand{\sitj}{(\sigma_j,T_j)}

For completeness, we now show that the lower bound shown above is tight even in the weighted case, in the absence of computational constraints. Recall that we are given a collection $\ess$ of subsets of a universe $U$, and a private subset $R \subseteq U$ of elements to be covered. Additionally, we have weights on sets; we round up weights to powers of 2, so that sets in $\ess_j$ have weight exactly $2^{-j}$. Without loss of generality, the largest weight is 1 and the smallest weight is $w=2^{-L}$.

As before, we will output a permutation $\pi$ on $\ess$, with the understanding that the cost $cost(R,\pi)$ of a permutation $\pi$ on input $R$ is defined to be the total cost of the set cover resulting from picking the first set in the permutation containing $e$, for each $e \in R$.

Our algorithm constructs this permutation in a gradual manner. It maintains a permutation $\pi_j$  on $\cup_{i\leq j} \ess_i$ and a threshold $T_j$. In step $j$, given $\pi_{j-1}$ and $T_{j-1}$, the algorithm constructs a partial permutation $\pi_j$ on $\cup_{i\leq j} \ess_j$, and a threshold $T_j$. In each step, we use the exponential mechanism to select an extension with an appropriate base distribution $\mu_j$ and score function $q$. At the end of step $L$, we get our permutation $\pi=\pi_L$ on $\ess$.

Our permutations $\pi_j$ will all have a specific structure. The weight of the $i$th set in the permutation, as a function of $i$ will be a unimodal function that is non-increasing until $T_j$, and then non-decreasing.  In other words, $\pi_j$ contains sets from $\ess_j$ as a continuous block. The sets that appear before $T_j$ are said to be in the {\em bucket}. We call a partial permutation respecting this structure {\em good}. We say a good permutation $\pi$ {\em extends} a good partial permutation $\pi_j$ if $\pi_j$ and $\pi$ agree on their ordering on $\cup_{i\leq j} \ess_i$.

We first define the score function that is used in these choices. A natural objective function would be $cost(R,\pi_j)= \min_{\pi \mbox{ extends } \pi_j} cost(R,\pi)$, i.e. the cost of the optimal solution conditioned on respecting the partial permutation $\pi_j$. We use a slight modification of this score function: we force the cover to contain all sets from the bucket and denote as $\mycost(R,\pi)$ the resulting cover defined by $R$ on $\pi$.  We then define $\mycost(R,\pi_j)$ naturally as $\min_{\pi \mbox{ extends } \pi_j} \mycost(R,\pi)$.   We first record the following easy facts:

\begin{observation}
\label{obs:wsc-ineff}
For any $R$, $\min_{\pi} \mycost(R,\pi) = \min_{\pi} cost(R,\pi) = \OPT$. Moreover, for any $\pi$, $\mycost(R,\pi)\geq cost(R,\pi)$.
\end{observation}

To get $(\pi_j,T_j)$ given $(\pi_{j-1},T_{j-1})$, we insert a permutation $\sigma_j$ of $\ess_j$ after the first $T_{j-1}$ elements of $\pi_{j-1}$, and choose $T_j$, where both $\sigma_j$ and $T_j$ are chosen using the exponential mechanism. The base measure on $\sigma_j$ is uniform and the base measure on $T_j-T_{j-1}$ is the geometric distribution with parameter $1/m^2$.

Let $\mycost(A,(\sigma_j,T_j))$ be defined as $\mycost(A,\pi_j)-\mycost(A,\pi_{j-1})$, where $\pi_j$ is constructed from $\pi_{j-1}$ and $(\sigma_j,T_j)$ as above. The score function we use to pick $\sitj$ is $score_j(R,\sitj) = 2^j \mycost(R,\sitj)$. Thus $Pr[(\sigma_j,T_j)] \propto (1/m^2(T_j-T_{j-1})) \exp(\eps score(\sitj))$.

Let the optimal solution to the instance contain $n_j$ sets from $\ess_j$. Thus $\OPT = \sum_{j} 2^{-j} n_j$. We first show that $\mycost(R,\pi_L)$ is $O(\OPT \log m/\eps)$. By Observation~\ref{obs:wsc-ineff}, the approximation guarantee would follow.

The probability that the $n_j$ sets in $\OPT$ fall in the bucket when picking from the base measure is at least $1/m^{3n_j}$. When that happens, $\mycost(R,\pi_j) = \mycost(R,\pi_{j-1})$. Thus the exponential mechanism ensures that except with probability $1/poly(m)$:
\begin{displaymath}
\mycost(R,\pi_j) \leq \mycost(R,\pi_{j-1}) + 4\cdot 2^{-j} \log (m^{3n_j})/\epsilon = \mycost(R,\pi_{j-1}) + 12\cdot 2^{-j}n_j \log m / \epsilon
\end{displaymath}

Thus with high probability,
\begin{eqnarray*}
\mycost(R,\pi_L) &\leq& \mycost(R,\pi_0) + 12 \sum_{j} 2^{-j} n_j \log m/\epsilon\\
&=& \OPT + 12 \OPT \log m /\epsilon
\end{eqnarray*}

Finally, we analyze the privacy.  Let $e\in U$ be an element such that the cheapest set covering $U$ has cost $2^{-j_e}$. Let $A$ and $B$ be two instances that differ in element $e$.  It is easy to see that $|\mycost(A,\sitj)-\mycost(B,\sitj)|$ is bounded by $2^{-j}$ for all $j$. We show something stronger:

\begin{lemma}
\label{lem:wsc-ineff-sens}
For any good partial permutation $\pi_j$ and any $A,B$ such that $A= B \cup \{e\}$,
\begin{displaymath}
|score(A,\sitj)-score_j(B,\sitj)| \leq
\left\{
\begin{array}{ll}
0 &\mbox{ if } j > j_e\\
2^{j_e-j+1} & \mbox{ if} j \leq j_e
\end{array}
\right.
\end{displaymath}
\end{lemma}
\begin{proof}
Let $\pi_B$ be the permutation realizing $\mycost(B)$. For $j \leq j_e$, if $e$ is covered by a set in the bucket in $\pi_B$, then the cost of $\pi_B$ is no larger in instance $A$ and hence $\mycost(A,\pi_j) = \mycost(B,\pi_j)$\footnote{We remark that this is not true for the function $cost$, and is the reason we had to modify it to $\mycost$.}. In the case that the bucket in $\pi_B$ does not cover $e$, then $\mycost(A,\pi_j) \leq \mycost(A,\pi_B) = \mycost(B,\pi_B) + 2^{-j_e} = \mycost(B,\pi_j) + 2^{-j_e}$. Since this also holds for $\pi_{j-1}$, this implies the claim from $j \leq j_e$.

For $j>j_e$, observe that the first set in $\pi_B$ that covers $e$ is fully determined by the partial permutation $\pi_j$, since the sets in $\cup_{i>j_e} \ess_i$ do not contain $e$. Thus $\mycost(A,\sitj) = \mycost(B,\sitj)$ and the claim follows.
\end{proof}

Then for any $j \leq j_e$, lemma~\ref{lem:wsc-ineff-sens} implies that for any $(\sigma_j,T_j)$,  $\exp(\epsilon(score(A,(\sigma_j,T_j))-score(B,(\sigma_j,T_j)))) \in [\exp(-\epsilon 2^{j_e-j+1}), \exp(\epsilon 2^{j_e-j+1})]$. Thus
\begin{eqnarray*}
\frac{\Pr[\sigma_j,T_j| A]}{\Pr[\sigma_j,T_j|B]} &\in & [\exp(-2^{j-j_e+2}\epsilon),\exp(2^{j-j_e+2}\epsilon)]
\end{eqnarray*}
Moreover, for any $j > j_e$, this ratio is 1. Thus
\begin{eqnarray*}
\frac{\Pr[\sigma_j,T_j| A]}{\Pr[\sigma_j,T_j|B]} &\in& [\Pi_{j\leq j_e}\exp(-2^{j-j_e+2}\epsilon)
,\Pi_{j\leq j_e}\exp(2^{j-j_e+2}\epsilon)]\\
&\subseteq& [\exp(-8\epsilon),\exp(8\epsilon)],
\end{eqnarray*}
which implies $8\epsilon$-differential privacy.

\section{Facility Location}

Consider the metric facility location problem: we are given a metric
space $(V,d)$, a facility cost $f$ and a (private) set of demand points
$D \subseteq V$. We want to select a set of facilities $F \subseteq V$
to minimize $\sum_{v \in D} d(v,F) + f\cdot |F|$. (Note that we assume
``uniform'' facility costs here instead of different costs $f_i$ for
different $i \in V$.) Assume that distances are at least $1$, and let
$\Delta = \max_{u,v} d(u,v)$ denote the diameter of the space.

We use the result of Fakcharoenphol et al.~\cite{FRT03} that any metric
space on $n$ points can be approximated by a distribution over
dominating trees with expected stretch $O(\log n)$; moreover all the
trees in the support of the distribution are rooted $2$-HSTs---they have
$L = O(\log \Delta)$ levels, with the leaves (at level $0$) being
exactly $= V$, the internal nodes being all Steiner nodes, the root
having level $L$, and all edges between levels~$(i+1)$ and~$i$ having
length $2^i$. Given such a tree $T$ and node $v$ at level $i$, let $T_v$
denote the (vertices in) the subtree rooted at $v$.

By \lref[Corollary]{cor:facloc-lbd}, it is clear that we cannot output
the actual set of facilities, so we will instead output instructions in
the form of an HST $T = (V_T, E_T)$ and a set of facilities $F \sse
V_T$: each demand $x \in D$ then gets assigned to its ancestor facility
at the lowest level in the tree. (We guarantee that the root is always
in $F$, hence this is well-defined.)  Now we are charged for the
connection costs, and for the \emph{facilities that have at least one
  demand assigned to them}.

\begin{algorithm}
  \caption{The Facility Location Algorithm}
  \begin{algorithmic}[1]
    \STATE \textbf{Input:} Metric $(V,d)$, facility cost $f$, demands $D
    \subseteq V$,$\epsilon$.
    \STATE Pick a random distance-preserving FRT tree $T$; recall this
    is a $2$-HST with $L = O(\log \Delta)$ levels.
    \STATE \textbf{let} $F \gets $ root $r$.
    \FOR{$i=1$ to $L$}
      \FOR{all vertices $v$ at level $i$}
        \STATE \textbf{let} $N_v = |D \cap T_v|$ and $\widetilde{N_v} = N_v +
        \mathrm{Lap}(L/\epsilon)$. \label{step-fl-noise}
        \STATE \textbf{if} $\widetilde{N_v} \cdot 2^i > f$ \textbf{then}
         $F \gets F \cup v$.
      \ENDFOR
    \ENDFOR
    \STATE \textbf{output} $(T, F)$: each demand $x \in D$ is assigned
    to the ancestor facility at lowest level in $T$.
  \end{algorithmic}
\end{algorithm}


\begin{theorem}
  \label{thm:facloc}
  The above algorithm preserves $\epsilon$-differential
  privacy and outputs a solution of cost $\OPT \cdot O(\log n \log \Delta) \cdot \frac{\log \Delta}{\epsilon}
  \log \left(\frac{n \log^2 \Delta}{\epsilon}\right)$.
\end{theorem}

For the privacy analysis, instead of outputting the set $F$ we could
imagine outputting the tree $T$ and all the counts $\widetilde{N_v}$;
this information clearly determines $F$. Note that the tree is
completely oblivious of the demand set. Since adding or removing
any particular demand vertex can only change $L$ counts, and the noise
added in Step~\ref{step-fl-noise} gives us $\eps/L$-differential
privacy, the fact that differential privacy composes linearly gives us
the privacy claim.

For the utility analysis, consider the ``noiseless'' version of the
algorithm which opens a facility at $v$ when $N_v \cdot 2^i \geq f$. It
can be shown that this ideal algorithm incurs cost at most $f + O(\log n
\log \Delta) \cdot \OPT$ (see, e.g.,~\cite[Theorem~3]{Indyk04}).  We
now have two additional sources of error due to the noise:
\begin{OneLiners}
\item Consider the case when $N_v\cdot2^i \geq f > \widetilde{N_v} \cdot
  2^i$, which increases the connection cost of some demands in $D$.
  However, the noise is symmetric, and so we overshoot the mark with
  probability at most $1/2$---and when this happens the $2$-HST property
  ensures that the connection cost for any demand $x$ increases by at
  most a factor of $2$.  Since there are at most $L = O(\log \Delta)$
  levels, the expected connection cost increases by at most a factor
  of~$L$.
\item Consider the other case when $N_v\cdot2^i < f \leq \widetilde{N_v}
  \cdot 2^i$, which increases the facility cost. Note that if $N_v \cdot
  2^i \geq f/2$, then opening a facility at $v$ can be charged again in
  the same way as for the noiseless algorithm (up to a factor of~$2$).
  Hence suppose that $\widetilde{N_v} - N_v \geq \frac12
  (f/2^i)$, and hence we need to consider the probability $p_i$ of the
  event that $\textrm{Lap}(L/\epsilon) > \frac12 (f/2^i)$, which is just
  $\frac{L}{\epsilon} \exp(- \frac{f}{2^{i+1}} \frac{\epsilon}{L})$.

  \medskip Note that if for some value of $i$, $f \geq
  \frac{L\,2^{i+1}}{\epsilon} \log \frac{L^2n}{\eps}$, the above
  probability $p_i$ is at most $1/Ln$, and hence the expected cost of
  opening up spurious facilities at nodes with such values of $i$ is at
  most $(1/Ln)\cdot Ln \cdot f = f$. (There are $L$ levels, and at most
  $n$ nodes at each level.)

  \medskip For the values of $i$ which are higher; i.e., for which $f <
  \frac{L\,2^{i+1}}{\epsilon} \log \frac{L^2n}{\eps}$, we pay for this
  facility only if there is a demand $x \in D$ in the subtree below $v$
  that actually uses this facility. Hence this demand $x$ must have used
  a facility above $v$ in the noiseless solution, and we can charge the
  cost $f$ of opening this facility to length of the edge $2^{i+1}$
  above $v$. Thus the total cost of spurious facilities we pay for is
  the cost of the noiseless solution times a factor $\frac{L}{\epsilon}
  \log \frac{L^2n}{\eps}$.
\end{OneLiners}
Thus the expected cost of the solution is at most
\begin{gather}
  \OPT \cdot O(\log n \log \Delta) \cdot \frac{\log \Delta}{\epsilon}
  \log \left(\frac{n \log^2 \Delta}{\epsilon}\right).
\end{gather}

\section{Combinatorial Public Projects (Submodular Maximization)}
\label{sec:cpp}

Recently Papadimitriou et al.\cite{PSS08} introduced the Combinatorial Public
Projects Problem (CPP Problem) and showed that there is a succinctly representable version of the problem for which, although there exists a constant factor approximation algorithm, no efficient \emph{truthful} algorithm can guarantee an approximation ratio better than $m^{\frac{1}{2}-\epsilon}$, unless $NP \subseteq BPP$. Here we adapt our set cover algorithm to give a privacy preserving approximation to the CPP problem within logarithmic (additive) factors.

In the CPP problem, we have $n$ agents and $m$ resources publicly
known. Each agent submits a private non-decreasing and \emph{submodular} valuation
function $f_i$ over subsets of resources, and our goal is to select a size-$k$
subset $S$ of the resources to maximize $\sum_{i=1}^nf_i(S)$. We assume that we have oracle access to the functions $f_i$. Note that since each $f_i$ is submodular, so is $\sum_{i=1}^nf_i(S)$, and
our goal is to produce a algorithm for submodular maximization that
preserves the privacy of the individual agent valuation
functions. Without loss of generality, we will scale the valuation
functions such that they take maximum value 1: $\max_{i,S}f_i(S) =
1$.

Once again, we have an easy computationally inefficient algorithm.
\begin{theorem}
The exponential mechanism when used to choose $k$ sets runs in time
$O({m \choose k} poly(n))$ and has expected quality at least $(1-1/e) OPT - O(\log {m
  \choose k}/\epsilon)$.
\end{theorem}

We next give a computationally efficient algorithm with slightly worse
guarantees. We adapt our unweighted set cover algorithm, simply
selecting $k$ items greedily:
\begin{algorithm}
  \caption{CPP Problem}
  \begin{algorithmic}[1]
    \STATE \textbf{Input:} A set of $M$ of $m$ resources, private
    functions $f_1,\ldots,f_n$, a number of resources $k$, $\epsilon,\delta$.
    \STATE \textbf{let} $M_1 \leftarrow M$, $F(x) := \sum_{i=1}^m
    f_i(x)$, $S_1 \leftarrow \emptyset$, $\epsilon' \leftarrow \frac{\epsilon}{e\ln (e/\delta)}$.
    \FOR{$i = 1$ to $k$}
    \STATE \textbf{pick} a resource $r$ from $M_i$ with probability
    proportional to $\exp(\epsilon' (F(S_i + \{r\}) - F(S_i)))$.
    \STATE \textbf{let} $M_{i+1} \leftarrow M_i - \{r\}$, $S_{i+1} \leftarrow S_i + \{r\}$.
    \ENDFOR
    \STATE \textbf{Output} $S_{k+1}$.
\end{algorithmic}
\end{algorithm}

\subsection{Utility Analysis}
\begin{theorem}
Except with probability $O(1/\textrm{poly}(n))$, the algorithm for the
CPP problem returns a solution with quality at least $(1-1/e) \OPT - O(k\log
m/\epsilon').$

\end{theorem}
\begin{proof}
Since $F$ is submodular and there exists a set $S^*$ with $|S| = k$
and $F(S) = \OPT$, there always exists a resource $r$ such that $F(S_i
+ \{r\})-F(S_i) \geq (\OPT-F(S_i))/k$. If we always selected the
optimizing resource, the distance to $\OPT$ would decrease by a factor
of $1-1/k$ each round, and we would achieve an approximation factor of
$1-1/e$. Instead, we use the exponential mechanism  which, by
\eqref{eqn:expmech}, selects a resource within $4\ln m/\epsilon'$ of
the optimizing resource with probability at least $1-1/m^3$. With
probability at least $1-k/m^3$ each of the $k$ selections decreases
$\OPT - F(S_i)$ by a factor of $(1-1/k)$, while increasing it by at
most an additive $4\ln m/\epsilon'$, giving $(1-1/e)\OPT + O(k\ln
m/\epsilon')$.
\end{proof}
\subsection{Privacy Analysis}
\begin{theorem}
For any $\delta \leq 1/2$, the CPP problem algorithm preserves
$(\epsilon'(e-1)\ln(e/\delta),\delta)$-differential privacy.
\end{theorem}
\begin{proof}
Let $A$ and $B$ be two $CPP$ instances that differ in a single agent $I$
with utility function $f_I$. We show that the output set of resources,
even revealing the
order in which the resources were chosen, is privacy preserving. Fix
some ordered set of $k$ resources, $\pi_1,\ldots,\pi_k$ write $S_i =
\bigcup_{j=1}^{i-1}\{\pi(j)\}$ to denote the first $i-1$ elements, and
write $s_{i,j}(A) = F_A(S_i + \{j\}) - F_A(S_i)$ to denote the
marginal utility of item $j$ at time $i$ in instance $A$. Define
$s_{i,j}(B)$ similarly for instance $B$.  We
consider the relative probability of our mechanism outputting ordering
$\pi$ when given inputs $A$ and $B$:
\ifnum\fullversion=1
\[ \frac{\Pr[M(A)=\pi]}{\Pr[M(B)=\pi]} =
 \prod_{i=1}^k\left(\frac{\exp(\epsilon'\cdot
     s_{i,\pi_i}(A))/(\sum_j\exp(\epsilon'\cdot
     s_{i,j}(A)))}{\exp(\epsilon'\cdot
     s_{i,\pi_i}(B))/(\sum_j\exp(\epsilon'\cdot s_{i,j}(B)))}\right),\]
\else
\begin{align*} &\frac{\Pr[M(A)=\pi]}{\Pr[M(B)=\pi]}\\ &=
 \prod_{i=1}^k\left(\frac{\exp(\epsilon'\cdot
     s_{i,\pi_i}(A))/(\sum_j\exp(\epsilon'\cdot
     s_{i,j}(A)))}{\exp(\epsilon'\cdot
     s_{i,\pi_i}(B))/(\sum_j\exp(\epsilon'\cdot s_{i,j}(B)))}\right),
     \end{align*}
\fi
where the sum over $j$ is over all remaining unselected resources.  We
can separate this into two products
\[ \prod_{i=1}^k\left(\frac{\exp(\epsilon'\cdot
       s_{i,\pi_i}(A))}{\exp(\epsilon'\cdot
       s_{i,\pi_i}(B))}\right)\cdot
   \prod_{i=1}^k\left(\frac{\sum_j\exp(\epsilon'\cdot
       s_{i,j}(B))}{\sum_j\exp(\epsilon'\cdot s_{i,j}(A))}\right).\]
If $A$ contains agent $I$ but $B$ does not, the second product is at
most 1, and the first is at most $\exp(\epsilon'\sum_{i=1}^k (F_I(S_i) - F_I(S_{i-1}))) \leq \exp(\epsilon')$.
If $B$ contains
agent $I$, and $A$ does
not, the first product is at most 1, and in the remainder of the
proof, we focus on this case.
We will write $\beta_{i,j} =
s_{i,j}(B)-s_{i,j}(A)$ to be the additional marginal utility of item
$j$ at time $i$ in instance $B$ over instance $A$, due to agent
$I$. Thus 
\begin{eqnarray*}
 \frac{\Pr[M(A)=\pi]}{\Pr[M(B)=\pi]}
&\leq& 
   \prod_{i=1}^k\left(\frac{\sum_j\exp(\epsilon'\cdot
       s_{i,j}(B))}{\sum_j\exp(\epsilon'\cdot s_{i,j}(A))}\right) \\
   &=&       
   \prod_{i=1}^k\left(\frac{\sum_j\exp(\epsilon'\beta_{i,j})\cdot\exp(\epsilon'\cdot
       s_{i,j}(A))}{\sum_j\exp(\epsilon'\cdot s_{i,j}(A))}\right) \\
   &=&   
\prod_{i=1}^k\E_i[\exp(\epsilon'\beta_{i})],
 \end{eqnarray*}
 where $\beta_i$ is the marginal utility actually achieved at time $i$
 by agent $I$, and the expectation is taken over the probability
 distribution over resources selected at time $i$ in instance $A$. For
 all $x \leq 1$, $e^x \leq 1 + (e-1)\cdot x$. Therefore, for all
 $\epsilon' \leq 1$, we have:
 \begin{eqnarray*}
 \prod_{i=1}^k\E_i[\exp(\epsilon'\beta_{i})] &\leq& \prod_{i=1}^kE_i[1
 + (e-1)\epsilon'\beta_i] \\
 &\leq& \exp((e-1)\epsilon'\sum_{i=1}^kE_i[\beta_i]).
 \end{eqnarray*}
As in the set-cover proof, we split the set of possible outputs into two sets. We call an output sequence {\em $q$-good} for an agent $I$ in instance $A$ if this sum $\sum_{i=1}^k E_i[\beta_i]$ is bounded above by $q$, and call it {\em $q$-bad} otherwise. For a $(\ln (e\delta^{-1}))$-good output $\pi$, we can then write
 $$\frac{\Pr[M(A)=\pi]}{\Pr[M(B)=\pi]} \leq
 \exp((e-1)\epsilon'\cdot\ln (e\delta^{-1}).$$
Moreover, note that since the total realized utility of any agent is at most $1$, if agent $I$ has realized utility $u_{i-1}$ before the $i$th set is chosen, then $\beta_i$ is distributed in $[0,1-u_{i-1}]$. Moreover, $u_{i} = u_{i-1} + \beta_i$. Lemma~\ref{lem:headsprobcont} then implies that the probability that the algorithms outputs a $(\ln (e\delta^{-1}))$-bad permutation is at most $\delta$. The theorem follows.
\end{proof}

\begin{remark} By choosing $\epsilon'=\epsilon/k$, we immediately get $\epsilon$-differential privacy and expected utility at least $(1-1/e) \OPT - O(k^2\ln m/\epsilon)$. This may give better guarantees for some values of $k$ and $\delta$.
\end{remark}

We remark that the $k$-coverage problem is a special case of the CPP
problem. Therefore:
\begin{corollary}
The CPP algorithm (with \emph{sets} as resources) is an
$(\epsilon,\delta)$-differential privacy preserving algorithm for the $k$-coverage problem achieving approximation factor at least $(1-1/e) \OPT - O(k\log m\log(2/\delta) /\epsilon).$
\end{corollary}

\subsection{Truthfulness}
The CPP problem can be viewed as a mechanism design problem when each
agent $i$ has a choice of whether to submit his actual
valuation function $f_i$, or to lie and submit a different valuation
function $f'_i$ if such a misrepresentation yields a better outcome
for agent $i$. A mechanism is \emph{truthful} if for every valuation
function of agents $j \ne i$, and every valuation function $f_i$ of
agent $i$, there is never a function $f'_i \ne f_i$ such that agent
$i$ can benefit by misrepresenting his valuation function as
$f'_i$. Intuitively, a mechanism is approximately truthful if no agent
can make more than a slight gain by not truthfully reporting.
\begin{definition}
A mechanism for the CPP problem is $\gamma$-truthful if for
every agent $i$, for every set of player valuations $f_j$ for $j \ne
i$, and for every valuation function $f'_i \ne f_i$:
\ifnum\fullversion=1
$$E[f_i(M(f_1,\ldots,f_i,\ldots,f_n))] \geq
E[f_i(M(f_1,\ldots,f'_i,\ldots,f_n))]-\gamma$$
\else
\begin{align*}
&E[f_i(M(f_1,\ldots,f_i,\ldots,f_n))] \\
&\geq E[f_i(M(f_1,\ldots,f'_i,\ldots,f_n))]-\gamma
\end{align*}
\fi
Note that $0$-truthfulness corresponds to the usual notion of
(exact) truthfulness.
\end{definition}
$(\epsilon,\delta)$-differential privacy in our setting immediately implies
$(2\epsilon+\delta)$-approximate truthfulness.
We note that Papadimitriou et al. \cite{PSS08} showed that
the CPP problem is inapproximable to an $m^{\frac{1}{2}-\epsilon}$ multiplicative factor
by any polynomial time $0$-truthful mechanism. Our result shows that relaxing that to $\gamma$-truthfulness allows us to give a constant approximation to the utility whenever $\OPT \geq 2k\log m \log (1/\gamma)/\gamma$ for any $\gamma$.

\subsection{Lower Bounds}

\begin{theorem}
 No $\epsilon$-differentially private algorithm for the maximum
 coverage problem can guarantee profit larger than $\OPT - (k\log
 (m/k)/20\epsilon)$.
\end{theorem}
The proof is almost identical to that of the lower bound
\lref[Theorem]{thm:kmed-lbd} for $k$-median, and hence is omitted.

\section{Steiner Forest}
\label{sec:steiner-forest}

Consider the Steiner network problem, where we are given a metric space
$M = (V, d)$ on $n$ points, and a (private) subset $R \subseteq V \times
V$ of source-sink (terminal) pairs. The goal is to buy a
minimum-cost set of edges $E(R) \subset \binom{V}{2}$ such that these edges
connect up each terminal pair in $R$. As in previous cases, we give
instructions in the form of a tree $T = (V,E_T)$; each terminal pair
$(u,v) \in R$ takes the unique path $P_T(u,v)$ in this tree $T$ between
themselves, and the (implicit) solution is the set of edges $E(R) =
\bigcup_{(u,v) \in R} P_T(u,v)$.

The tree $T$ is given by the randomized construction of Fakcharoenphol
et al.~\cite{FRT03}, which guarantees that $\E[\cost(E(R))] \leq O(\log n)
\cdot \OPT$; moreover, since the construction is oblivious to the set
$R$, it preserves the privacy of the terminal pairs perfectly (i.e.,
$\epsilon = 0$). The same idea can be used for a variety of network
design problem (such as the ``buy-at-bulk'' problem) which can be solved
by reducing it to a tree instance.

\section{Private Amplification Theorem}
\label{sec:boosting}

\ktnote{Change name to amplification? -- not sure about this. TBD.}
\ktnote{Compare with direct repetition. Done. --k}
In this section, we show that differentially private mechanisms that
give good guarantees in expectation can be repeated privately to amplify
the probability of a good outcome. First note that if we simply repeat a private algorithm $T$ times, and select the best outcome, we can get the following result:
\begin{theorem}
Let $M: D \rightarrow R$ be an $\epsilon$-differentially private mechanism such that for a query function $q$, and a parameter $Q$, $\Pr[q(A,M(A)) \geq Q] \geq \frac{1}{2}$. Then for any $\delta>0$, $\epsilon' \in (0,\frac{1}{2})$, there is a mechanism $M'$ which satisfies the following properties:
\begin{OneLiners}
\item \textbf{Utility:} $\Pr[q(A,M(A)) \geq Q] \geq (1-2^{-T})$.
\item \textbf{Efficiency:} $M'$ makes  $T$ calls to $M$.
\item \textbf{Privacy:} $M'$ satisfies $(\epsilon T)$-differential privacy.
\end{OneLiners}
\end{theorem}

Note that the privacy parameter degrades linearly with $T$. Thus to bring down the failure probability to inverse polynomial, one will have to make $T$ logarithmic. To get $\epsilon'$-differential privacy, one would then take $\epsilon$ to be $\epsilon'/T$. If $Q$ was inversely proportional to $\epsilon$, as is the case in many of our algorithms, this leads to an additional logarithmic loss. The next theorem shows a more sophisticated amplification technique that does better.

\begin{theorem}[Private Amplification Theorem]
  Let $M: D \rightarrow R$ be an $\epsilon$-differentially private
  mechanism such that for a query function $q$ with sensitivity~$1$, and
  a parameter $Q$, $\Pr[q(A,M(A)) \geq Q] \geq p$ for some $p \in
  (0,1)$. Then for any $\delta>0$, $\epsilon' \in (0,\frac{1}{2})$,
  there is a mechanism $M'$ which satisfies the following properties:
\begin{OneLiners}
\item 
$\Pr[q(A,M(A)) \geq Q - \frac4{\epsilon'}\log(\frac{1}{\epsilon' \delta p})] \geq (1-\delta)$.
\item 
$M'$ makes  $O((\frac{1}{\epsilon' \delta p})^2\log(\frac{1}{\epsilon'\delta p}))$ calls to $M$.
\item 
$M'$ satisfies $(\epsilon+8\epsilon')$-differential privacy.
\end{OneLiners}
\end{theorem}

\begin{proof}
  Let $T=(\frac{8}{\epsilon' \delta p})^2\log(\frac{1}{\epsilon' \delta
    p})$. The mechanism $M'$ runs $M$ on the input $A$ independently
  $(T+1)$ times to get outputs $S_1=\{r_1,\ldots,r_{T+1}\}$. It also
  adds in $T'=\frac{\sqrt{4T\log T}}{\epsilon'}$ dummy outcomes
  $S_2=\{s_1,\ldots,s_{T'}\}$ and selects an outcome from $S_1 \cup S_2$
  using the exponential mechanism with privacy parameter $\eps'$ and
  score function
  \begin{displaymath}
    \widetilde{q}(A,r) = \left\{\begin{array}{ll}
        \min(Q,q(A,r)) & \mbox{if } r\in S_1\\
        Q & \mbox{if } r \in S_2
      \end{array} \right.
  \end{displaymath}

  The efficiency of $M'$ is immediate from the construction. To analyze
  the utility, note that~\eqref{eqn:expmech} ensures that the
  exponential mechanism's output $r$ satisfies $\widetilde{q}(A,r) > Q -
  \frac{4}{\epsilon'}\log(\frac{1}{\epsilon' \delta p})$ with
  probability $(1-\frac{\delta}{2})$. Conditioned on the output $r$
  satisfying this property, the ratio $\Pr[r\in S_1]/\Pr[r \in S_2]$ is at
  least $|\{r \in S_1: q(A,r)\geq Q\}|/|S_2|$. Since the numerator is at
  least $pT$ in expectation, the probability of $r$ being a dummy
  outcome is at most $\frac{\delta}{2}$. This establishes the utility
  property.

  We now show the privacy property. For any $r_0 \in R$,
  \begin{align}
    \Pr[M'(A)=r_0] &= \sum_{i=1}^{T+1} \Pr[r_i=r_0]
    \E\left[\frac{\exp(\epsilon'\widetilde{q}(A,r_0))}{\sum_{r\in
        S_1}\exp(\epsilon'\widetilde{q}(A,r))+ T'\exp(\epsilon' Q)} \,\bigg|\,
  r_i = r_0\right] \notag \\
    &= (T+1)\cdot \Pr[M(A) \mbox{ outputs }r_0] \cdot
    \exp(\epsilon'\widetilde{q}(A,r_0)) \cdot \E\left[\frac{1}{\sum_{r\in
        S_1}\exp(\epsilon'\widetilde{q}(A,r))+ T'\exp(\epsilon' Q)} \,\bigg|\,
  r_{T+1} = r_0\right] \notag \\
    &= (T+1)\cdot \Pr[M(A) \mbox{ outputs }r_0] \cdot
    \exp(\epsilon'\widetilde{q}(A,r_0)) \cdot \exp(-\epsilon'Q)\cdot
    \notag \\
    &~~~~~~~~~\E\left[\frac{1}{\sum_{r\in S_1\setminus
          \{r_0\}}\exp(\epsilon'(\widetilde{q}(A,r)-Q))+
        \exp(\epsilon'(\widetilde{q}(A,r_0)-Q))+ T'}\right] \label{par-eq1}
  \end{align}
where the expectation is also taken over runs $1,\ldots,T$ of $M$ (we've
explicitly conditioned on run $(T+1)$ producing $r_0$).

It is easy to bound the change in the first two terms when we change
from input $A$ to a neighboring input $B$, since $M$ satisfies
$\epsilon$-differential privacy, and $\widetilde{q}$ has sensitivity
1. Let $D = D(A)$ denote the denominator in the final expectation; we
would like to show that $\E[\frac{1}{D(A)}] \leq
\exp(\epsilon)\E[\frac{1}{D(B)}]$ for neighboring inputs $A$ and
$B$. Let $C=\exp(\epsilon'(\widetilde{q}(A,r_0)-Q))+ T'$ denote the
constant term in $D(A)$.

First observe that
\begin{eqnarray*}
\E[D(A)]&=&C+T \cdot \E_{r\in M(A)}[\exp(\epsilon'(\widetilde{q}(A,r)-Q)] \\
&\geq& C+T\cdot \exp(-\epsilon') \cdot \E_{r\in M(A)}[\exp(\epsilon'(\widetilde{q}(B,r)-Q)]\\
&\geq& C+T\cdot \exp(-2\epsilon')\cdot \E_{r\in M(B)}[\exp(\epsilon'(\widetilde{q}(B,r)-Q)]\\
&\geq& \exp(-2\epsilon') \cdot \E[D(B)],
\end{eqnarray*}
where the first inequality follows from the sensitivity of $q$ and the
second from the $\epsilon$-differential privacy of $M$. Thus $\E[D(A)]$
is close to $\E[D(B)]$. We now show that $\E[\frac{1}{D(A)}]$ is close
to $\frac{1}{\E[D(A)]}$ for each $A$, which will complete the proof.

The first step is to establish that $D(A)$ is concentrated around its
expectation. Since $D=C+\sum_{i=1}^T Y_i$, where the $Y_i$'s are
i.i.d.~random variables in $[0,1]$, standard concentration bounds imply
\begin{equation*}
\Pr[D \geq \E[D] + t] \leq \exp(-2t^2/T); \;\;\;\;\;\;\;\;\Pr[D \leq \E[D] - t] \leq \exp(-2t^2/T);
\end{equation*}

Since $\frac{1}{D} \geq \frac{1}{C}$, we can now estimate
\begin{eqnarray*}
\E[\frac{1}{D}] &\leq& \frac{\exp(\epsilon')}{\E[D]} + \int_{\frac{\exp(\epsilon')}{\E[D]}}^{\frac{1}{C}} \Pr[\frac{1}{D} \geq y] \ud y\\
&\leq& \frac{\exp(\epsilon')}{\E[D]} + \int_{C}^{\exp(-\epsilon')\E[D]} \frac{\Pr[D \leq z]}{z^2} \ud z\\
&\leq& \frac{\exp(\epsilon')}{\E[D]} + \frac{1}{C^2}\int_{C}^{\exp(-\epsilon')\E[D]} \exp(-2(z-\E[D])^2/T) \ud z\\
&\leq& \frac{\exp(\epsilon')}{\E[D]} + \frac{(\exp(-\epsilon')\E[D]-C)}{C^2} \exp(-(\epsilon' \E[D])^2/T)\\
&\leq& \frac{\exp(\epsilon')}{\E[D]} + \frac{1}{T^2}
\end{eqnarray*}
since $\E[D] > C > \frac{\sqrt{4T\log T}}{\epsilon'}$, $\E[D] < 2T$, and $C>1$. Thus $\E[\frac{1}{D}] \leq \frac{\exp(2\eps')}{\E[D]}$.

Similarly,
\begin{eqnarray*}
\E[\frac{1}{D}] &\geq& \frac{\exp(-\epsilon')}{\E[D]} -  \int_{0}^{\frac{\exp(-\epsilon')}{\E[D]}} \Pr[\frac{1}{D} \leq y] \ud y\\
&\geq& \frac{\exp(-\epsilon')}{\E[D]} -  \int_{\exp(\epsilon')\E[D]}^{\infty} \frac{\Pr[D \geq z]}{z^2} \ud z\\
&\geq& \frac{\exp(-\epsilon')}{\E[D]} -  \frac{\exp(-2\epsilon')}{\E[D]^2}\int_{\exp(\epsilon')\E[D]}^{\infty} \exp(-2(z-\E[D])^2/T) \ud z\\
&\geq& \frac{\exp(-\epsilon')}{\E[D]} -  \frac{\exp(-2\epsilon')}{\E[D]^2}\sqrt{T}\\
&\geq& \frac{\exp(-\epsilon')}{\E[D]} - \frac{\epsilon'}{\E[D]},
\end{eqnarray*}
so that $\E[\frac{1}{D}] \geq \frac{\exp(-3\epsilon')}{\E[D]}$.

Thus $\E[\frac{1}{D(A)}]\leq \exp(7\epsilon')\E[\frac{1}{D(B)}]$ for
neighboring inputs $A$ and $B$. Now using this fact in
expression~\eqref{par-eq1} for $\Pr[M'(A)=r_0]$ above, we conclude that
$M'$ satisfies $(\epsilon + 8\epsilon')$-differential privacy.
\end{proof}

\bibliographystyle{alpha}
\newcommand{\etalchar}[1]{$^{#1}$}

\appendix

\section{Unweighted Vertex Cover Algorithm: An Alternate View}
\label{sec:aaron-proof}

In this section, we consider a slightly different way to implement the
vertex cover algorithm.  Given a graph $G = (V,E)$, we mimic the
randomized proportional-to-degree algorithm for $\alpha n$ rounds
($\alpha < 1$), and output the remaining vertices in random order.  That
is, in each of the first $\alpha n$ rounds, we select the next vertex
$i$ with probability proportional to $d(i) + 1/\epsilon$: this is
equivalent to imagining that each vertex has $1/\epsilon$
``hallucinated'' edges in addition to its real edges. (It is most
convenient to imagine the other endpoint of these hallucinated edges as
being fake vertices which are always ignored by the algorithm.)

When we select a vertex, we remove it from the graph, together with the
real and hallucinated edges adjacent to it.  This is equivalent to
picking a random (real or hallucinated) edge from the graph, and
outputting a random real endpoint. Outputting a vertex affects the real
edges in the remaining graph, but does not change the hallucinated edges
incident to other vertices.

\paragraph{Privacy Analysis.}
The privacy analysis is similar to that of Theorem~\ref{thm:privacy}:
imagine the weights being $w_i = 1/\epsilon$ for the first $\alpha n$
rounds and $w_i = \infty$ for the remaining rounds, which gives us
$2\sum_{i = (1-\alpha) n}^n \frac{1}{i w_i} \leq \epsilon \;
(\frac{2\alpha}{1 - \alpha})$-differential privacy.

\paragraph{Utility Analysis.}
To analyze the utility, we couple our algorithm with a run of the
non-private algorithm $\mathcal{A}$ that at each step picks an arbitrary
edge of the graph and then picks a random endpoint: it is an easy
exercise that this an $2$-approximation algorithm.

We refer to vertices that have non-zero ``real'' degree at the time they
are selected by our algorithm as \textit{interesting vertices}: the cost
of our algorithm is simply the number of interesting vertices it selects
in the course of its run. Let $I_1$ denote the number of interesting
vertices it selects during the first $\alpha n$ steps, and $I_2$ denote
the number of interesting vertices it selects during its remaining
$(1-\alpha)n$ steps, when it is simply ordering vertices
randomly. Clearly, the total cost is $I_1 + I_2$.

We may view the first phase of our algorithm as selecting an edge at
random (from among both real and hallucinated ones) and then outputting
one of its endpoints at random. Now, for the rounds in which our
algorithm selects a real edge, we can couple this selection with one
step of an imagined run of $\mathcal{A}$ (selecting the same edge and
endpoint).  Note that this run of $\mathcal{A}$ maintains a vertex cover
that is a subset of our vertex cover, and that once our algorithm has
completed a vertex cover, no interesting vertices remain. Therefore,
while our algorithm continues to incur cost, $\mathcal{A}$ has not yet
found a vertex cover.

In the first phase of our algorithm, every interesting vertex our
algorithm selects has at least one real edge adjacent to it, as well as
$1/\epsilon$ hallucinated edges. Conditioned on selecting an interesting
vertex, our algorithm had selected a real edge with probability at least
$\epsilon' = 1/(1 + 1/\epsilon)$. Let $R$ denote the random variable
that represents the number of steps $\mathcal{A}$ is run for. $E[R] \leq
2\OPT$ since $\mathcal{A}$ is a $2$-approximation algorithm. By
linearity of expectation:
\begin{equation}
  \label{OPTCOST}
  2\OPT \geq \E[R] \geq \epsilon'\cdot E[I_1]
\end{equation}
We now show that most of our algorithm's cost comes from the first
phase, and hence that $I_2$ is not much larger than $I_1$.
\begin{lemma}
  \[\E[I_1] \geq \ln \left(\frac{1}{1-\alpha}\right) \cdot \E[I_2]\]
\end{lemma}
\begin{proof}
  Consider each of the $\alpha n$ steps of the first phase of our
  algorithm. Let $n_i$ denote the number of interesting vertices
  remaining at step $i$. Note that $\{n_i\}$ is a non-increasing
  sequence. At step $i$, there are $n_i$ interesting vertices and
  $n-i+1$ remaining vertices. Note that the probability of picking an
  interesting vertex is strictly greater than $n_i/(n-i+1)$ at each
  step. We may therefore bound the expected number of interesting
  vertices picked in the first phase:
  $$\E[I_1] > \sum_{i=1}^{\alpha n}\frac{\E[n_i]}{n-i+1} \geq \E[n_{\alpha
    n}]\sum_{j=(1-\alpha)n}^{n}\frac{1}{j} \geq \ln
  \left(\frac{1}{1-\alpha}\right) \cdot \E[n_{\alpha n}]$$
  Noting that $\E[I_2] \leq \E[n_{\alpha n}]$ completes the proof.
\end{proof}
Combining the facts above, we get that
\begin{gather}
  \frac{\E[\mathsf{cost}]}{\OPT} \leq \frac{2}{\epsilon'} \; \left(1 +
    \frac{1}{\ln (1-\alpha)^{-1}}\right).
\end{gather}

\section{Missing Proofs}
\newcommand{\dee}{\mathcal{D}}

\label{sec:deltafix}
In this section, we prove Lemma~\ref{lem:deltafix}. The lemma is a
consequence of the following more general inequality.

Consider the following $n$ round probabilistic process. In each round,
an adversary chooses a $p_i \in [0,1]$ possibly based on the first
$(i-1)$ rounds and a coin is tossed with heads probability $p_i$. Let
$Z_i$ be the indicator for the the event that no coin comes up heads
in the first $i$ steps. Let $Y_j$ denote the random variable
$\sum_{i=j}^n p_i Z_i$ and let $Y = Y_1$.
\begin{lemma}
\label{lem:headsprob}
Let $Y$ be defined as above. Then for any $q$, $\Pr[Y > q] \leq \exp(-q)$.
\end{lemma}
\begin{proof}
  We claim that for any $j$ and any $q$, $\Pr[Y_j > q] \leq \exp(-q)$,
  which implies the lemma.  The proof is by reverse induction on
  $j$. For $j=n$, $Y_n$ is $0$ if the $n$th coin or any coin before it
  comes up heads and $p_n$
  otherwise. Thus for $q \geq p_n$, the left hand side is zero. For $q
  \in [0,p_n)$, the left hand side is at most $(1-p_n) \leq \exp(-p_n)
  \leq \exp(-q)$. Finally, for $q < 0$ the right hand side exceeds 1.

  Now suppose that for any adversary's strategy and for all $q$,
  $\Pr[Y_{j+1} > q] \leq \exp(-q)$. We will show the claim for
  $Y_j$. Once again, for $q \leq 0$, the claim is trivial. In round
  $j$, if the adversary chooses $p_j$, there is a probability $p_j$
  that the coin comes up heads so that $Y_j=0$. Thus for any $q \geq
  0$,  $\Pr[Y_j > q] = \Pr[p_j Z_j + Y_{j+1} > q] = (1-p_j)
  \Pr[Y_{j+1} > q-p_j]$. Using the inequality $(1-x)\leq \exp(-x)$ and
  the inductive hypothesis, the claim follows for $Y_j$.
\end{proof}

To map the randomized algorithm to the setting of
lemma~\ref{lem:headsprob}, we consider running the randomized weighted
set cover algorithm as follows. When choosing a set $S$ in step $i$,
the algorithm first tosses a coin whose heads probability is $p_i(A)$
to decide whether to pick a set covering $I$ or not. Then it uses a
second source of randomness to determine the set $S$ itself, sampling
from $\{S: I \in S\}$ or $\{S: I \not\in S\}$ with the appropriate
conditional probabilities based on the outcome of the coin.  Clearly
this is a valid implementation of the weighted set cover
algorithm. Note that the probabilities $p_i(A)$ may depend on the
actual sets chosen in the first $(i-1)$ steps if none of the first
$(i-1)$ coins come up heads. Since lemma~\ref{lem:headsprob} applies
even when $p_i(A)$'s are chosen adversarially,
lemma~\ref{lem:deltafix} follows.

We also prove a more general version of Lemma~\ref{lem:headsprob} that applies to non-Bernoulli distributions. This lemma will be needed to prove the privacy of our algorithm for submodular minimization in Section~\ref{sec:cpp}.
We now consider a different $n$ round probabilistic process. In each round,
an adversary chooses a distribution $\dee_i$ over $[0,1]$, possibly based on the first
$(i-1)$ rounds and a sample $R_i$ is drawn from the distribution $\dee_i$. Let $Z_0 = 1$ and let $Z_{i+1} = Z_i - R_i Z_i$. Let $Y_j$ denote the random variable $\sum_{j=1}^n Z_i E[R_i]$ and let $Y$ denote $Y_1$.

\begin{lemma}
\label{lem:headsprobcont}
Let $Y$ be defined as above. Then for any $q$, $\Pr[Y > q] \leq e\exp(-q)$.
\end{lemma}
\begin{proof}
We prove a stronger claim. We show that for $\Pr[Y_j \geq qZ_j] \leq e\exp(-q)$. The proof is by reverse induction on $j$. For $j=n$, $Y_n = E[R_n] Z_n \leq Z_n$ since $\dee_n$ is supported on $[0,1]$ and hence has expectation at most 1. Thus the claim is trivial for any $q \geq 1$. For $q\leq 1$, the right hand side is at least  1 and there is nothing to prove.
Supppose that for any $q$ and any strategy of the adversary, $\Pr[Y_{j+1} \geq qZ_{j+1}] \leq e\exp(-q)$. We show the claim for $Y_j$. Once again the case $q \leq 1$ is trivial, so we assume $q \geq 1$. Let $\mu_j$ denote $E[R_j]$. Note that $Y_{j} = Z_j \mu_j + Y_{j+1}$. Moreover, $Z_{j+1} = (1-R_j) Z_j$. Thus,
\ifnum\fullversion=1
$$ \Pr[Y_j \geq q Z_j] = E_{R_j \in \dee_j} [ \Pr[Y_{j+1} \geq qZ_j - \mu_j Z_j]] = E_{R_j \in \dee_j} [ \Pr[Y_{j+1} \geq \frac{q-\mu_j}{1-R_j} Z_{j+1}]] \leq E_{R_j \in \dee_j}[e \exp(-\frac{q-\mu_j}{1-R_j})].$$
\else
\begin{eqnarray*}
\Pr[Y_j \geq q Z_j] &=& E_{R_j \in \dee_j} [ \Pr[Y_{j+1} \geq qZ_j - \mu_j Z_j]]\\ &=& E_{R_j \in \dee_j} [ \Pr[Y_{j+1} \geq \frac{q-\mu_j}{1-R_j} Z_{j+1}]] \\ &\leq& E_{R_j \in \dee_j}[e \exp(-\frac{q-\mu_j}{1-R_j})].
\end{eqnarray*}
\fi
We show that for any distribution $\dee$, the last term is bounded by $e \exp(-q)$, which will complete the proof. Re-arranging, it suffices to show that for any distribution $\dee$ on $[0,1]$,
$$ E_{R \in \dee}[\exp(\frac{\mu - qR}{1-R})] \leq 1.$$
Since $\frac{\mu-qR}{1-R}$ is positive when $R \leq \mu/q$ and negative otherwise, one can verify that for any $R$,
$\exp(\frac{\mu - qR}{1-R}) \leq \exp(\frac{\mu - qR}{1-\frac{\mu}{q}})$. Moreover, since $\exp(\cdot)$ is convex, the function lies below the chord and we can conclude that $\exp(\frac{\mu - qR}{1-\frac{\mu}{q}}) \leq \exp(\frac{\mu}{1-\frac{\mu}{q}})  + R (\exp(\frac{\mu-q}{1-\frac{\mu}{q}}) - \exp(\frac{\mu}{1-\frac{\mu}{q}}))$. Thus it suffices to prove that
$$\exp(\frac{\mu}{1-\frac{\mu}{q}})  + \mu (\exp(\frac{\mu-q}{1-\frac{\mu}{q}}) - \exp(\frac{\mu}{1-\frac{\mu}{q}})) \leq 1,$$
or equivalently
$$1  + \mu (\exp(\frac{-q}{1-\frac{\mu}{q}}) - 1 \leq \exp(\frac{-\mu}{1-\frac{\mu}{q}}).$$
This rearranges to
$$ 1 - \exp(-\frac{\mu}{1-\frac{\mu}{q}})  \leq \mu (1 - \exp(-\frac{q}{1-\frac{\mu}{q}})).$$
Consider the function $f(x) = 1 - \exp(-\frac{x}{1-\frac{\mu}{q}})$. $f$ is convex with $f(0) = 0$ and $f(1) \leq f(q) = (1 - \exp(-\frac{q}{1-\frac{\mu}{q}}))$. Thus $f(\mu) \leq \mu f(1) \leq \mu f(q)$, for $q \geq 1$. The claim follows.
\end{proof}

\end{document}